\documentclass[12pt]{article}
\usepackage{mathtools}
\usepackage{amssymb}
\usepackage{geometry}
\usepackage[utf8]{inputenc}
\usepackage[english]{babel}
\usepackage{pict2e}
\usepackage{mathrsfs}
\usepackage[framemethod=tikz]{mdframed}
\usepackage{mathdots}
\usepackage{tikz}
\usepackage{lipsum}

\definecolor{mycolor}{rgb}{0.9, 0.1, 0.2}

\newmdenv[innerlinewidth=1.5pt, roundcorner=4pt,linecolor=mycolor,innerleftmargin=6pt,
innerrightmargin=6pt,innertopmargin=6pt,innerbottommargin=6pt]{mybox}
\usepackage{graphicx}
\usepackage{caption}
\usepackage{tikz,tikz-cd}
\usetikzlibrary{matrix, arrows}

\graphicspath{{/Users/Charley/Pictures/LaTeX/}}
\usepackage{letltxmacro}

\usepackage{xcolor}
\definecolor{green(munsell)}{rgb}{0.0, 0.66, 0.47}
\definecolor{BlueGreenn}{rgb}{0.3,0.5,0.8}
\definecolor{darkblue}{rgb}{0.1,0.12,0.24}
\definecolor{DB}{rgb}{0.3,0.3,0.3}
\definecolor{DOr}{rgb}{0.7,0.3,0.3}
\definecolor{DGr}{rgb}{0.3,0.7,0.3}
\definecolor{DBl}{rgb}{0.1,0.3,0.5}
\definecolor{arylideyellow}{rgb}{0.91, 0.84, 0.42}
\definecolor{burntorange}{rgb}{0.8, 0.33, 0.0}
\definecolor{chromeyellow}{rgb}{1.0, 0.65, 0.0}
\usepackage[hypertexnames=true, citecolor = burntorange, colorlinks = true, linkcolor = darkblue, pdffitwindow = false, urlbordercolor = white,pdfborder={0 0 0}]{hyperref}

\usepackage{csquotes}
\usepackage{amsthm}
\usepackage{amsmath}
\newcommand{\email}[1]{\href{mailto:#1}{#1}}
\usepackage{color,soul}
\usepackage{relsize}
\theoremstyle{plain}
\newtheorem{theorem}{Theorem}
\theoremstyle{definition}
\newtheorem{definition}[theorem]{Definition}
\theoremstyle{plain}
\newtheorem{corollary}[theorem]{Corollary}
\theoremstyle{plain}
\newtheorem{proposition}[theorem]{Proposition}
\theoremstyle{remark}

\theoremstyle{remark}
\newtheorem{remark}[theorem]{Remark}
\usepackage[rightcaption]{sidecap}
\newtheorem{example}[theorem]{Example}
\newtheorem{examples}[theorem]{Examples}
\theoremstyle{plain}
\newtheorem{lemma}[theorem]{Lemma}
\usepackage{enumerate}

\usepackage {tikz}
\usetikzlibrary {positioning}
\usepackage {xcolor}
\definecolor {processblue}{cmyk}{0.96,0,0,0}
\usetikzlibrary{positioning}

\newcommand{\bigbplus}{
  \mathop{
    \vphantom{\bigoplus} 
    \mathchoice
      {\vcenter{\hbox{\resizebox{\widthof{$\displaystyle\bigoplus$}}{!}{$\boxplus$}}}}
      {\vcenter{\hbox{\resizebox{\widthof{$\bigoplus$}}{!}{$\boxplus$}}}}
      {\vcenter{\hbox{\resizebox{\widthof{$\scriptstyle\oplus$}}{!}{$\boxplus$}}}}
      {\vcenter{\hbox{\resizebox{\widthof{$\scriptscriptstyle\oplus$}}{!}{$\boxplus$}}}}
  }\displaylimits 
}

\newcommand{\Hil}{{\mathcal H}}

\newcommand{\F}{{\mathcal F}}
\newcommand{\R}{{\mathcal R}}

\title{Fock representations of Zamolodchikov algebras\newline and R-matrices}
\author{Gandalf Lechner\thanks{Cardiff University, School of Mathematics,
Cardiff, CF24 4AG, UK.
E-mail: \email{LechnerG@Cardiff.ac.uk}}
\and
Charley Scotford\thanks{Cardiff University, School of Mathematics,
Cardiff, CF24 4AG, UK.
E-mail: \email{ScotfordC@Cardiff.ac.uk}}}
\date{September 28, 2019}
\begin{document}
\maketitle
\begin{abstract}
A variation of the Zamolodchikov-Faddeev algebra over a finite dimensional Hilbert space $\mathcal{H}$ and an involutive unitary $R$-Matrix $S$ is studied. This algebra carries a natural vacuum state, and the corresponding Fock representation spaces $\mathcal{F}_S(\mathcal{H})$ are shown to satisfy $\mathcal{F}_{S\boxplus R}(\mathcal{H}\oplus\mathcal{K}) \cong \mathcal{F}_S(\mathcal{H})\otimes \mathcal{F}_R(\mathcal{K})$, where $S\boxplus R$ is the box-sum of $S$ (on $\mathcal{H}\otimes\mathcal{H}$) and $R$ (on $\mathcal{K}\otimes\mathcal{K}$). This analysis generalises the well-known structure of Bose/Fermi Fock spaces and a recent result of Pennig.\par
It is also discussed to which extent the Fock representation depends on the underlying $R$-matrix, and applications to quantum field theory (scaling limits of integrable models) are sketched.
\end{abstract}

\section{Introduction}

The Zamolodchikov-Faddeev algebras \cite{ZamFadd,ZamFadd2} are a class of quadratic exchange algebras of ``creation'' and ``annihilation'' operators which generalize the familiar CCR and CAR algebras \cite{BratRob2} and are closely related to Wick algebras that allow normal ordering \cite{ZamAlg2}. These algebras and their Hilbert space representations are of central importance in integrable quantum field theory (see, for example, \cite{Smirnov:1992,AbdallaAbdallaRothe:2001,Lechner}) as well as in other fields such as $q$-deformations \cite{QDef,QDef5} or anyonic statistics.

The central datum defining the relations in a Zamolodchikov algebra is a solution $S$ of the Yang-Baxter equation. In the QFT context, $S$ plays the role of an elastic two-particle scattering matrix and depends on a spectral parameter (rapidity), i.e. it is a matrix-valued function ${\mathbb R}\ni\theta\mapsto S(\theta)\in\text{End}({\mathcal H}\otimes\mathcal H)$, where $\mathcal H$ is a finite-dimensional Hilbert space.

The present work has its background in an ongoing research project about short-distance scaling limits of integrable quantum field theories \cite{Me}, generalizing previous work on scalar models \cite{Scaling}. At finite scale, such a QFT is defined in terms of a mass parameter and a Fock (vacuum) representation of the Zamolodchikov algebra of its two-particle S-matrix~$S$. Essential properties of the scaling limit theory are governed by $\theta$-independent solutions of the Yang-Baxter equation derived from $S$, namely the value $S_0:=S(0)$ at zero rapidity transfer, and the two limits\footnote{The significance of these objects is explained in Section~\ref{section:outlook}.} $S_\pm:=\lim_{\theta\to\pm\infty}S(\theta)$.

This motivates a closer study of the Zamolodchikov algebras build from constant (parameter-independent) unitary solutions to the Yang-Baxter equation (``R-matrices''). As $R:=S_0,S_\pm$ are often involutive (that is, $R^2=1$), we can use and apply recently established tools and results about the structure of the space of all unitary involutive R-matrices \cite{Y-B}, denoted $\R_0(\Hil)$ for base space $\Hil$.

In Section 2, we define Zamolodchikov-Faddeev algebras as abstract unital ${}^*$-algebras and their natural vacuum functional $\omega$. Section 3 is then devoted to an analysis of Fock representations. These amount to carrying out the GNS construction w.r.t. $\omega$, but since it is not initially clear if $\omega$ is positive, an independent construction has to be given. We do this by adapting a concrete representation known from quantum field theory to our setting and verify the GNS property afterwards (Theorem~\ref{theorem:GNS}). Given the Zamolodchikov algebra based on an involutive R-matrix $S\in\R_0(\Hil)$, this provides us with an $S$-symmetric Fock representation space $\F_S(\Hil)$. For special choices of $S$, this coincides with the Bose or Fermi Fock space over~$\Hil$.

In Sections 4 and 5, we study the dependence of $\F_S(\Hil)$ on $S$. Adopting the box-sum operation $\boxplus$ on $\bigcup_\Hil\R_0(\Hil)$ from \cite{Y-B} (which yields $R\boxplus S\in\R_0(\Hil\oplus{\mathcal K})$ for $R\in\R_0(\Hil)$, $S\in\R_0({\mathcal K})$), we prove 
\begin{align}\label{1}
 \F_{R\boxplus S}(\Hil\oplus{\mathcal K})
 \cong
 \F_R(\Hil) \otimes \F_S({\mathcal K})
\end{align}
in Theorem~\ref{MainThe} as our main result. It is interesting to note that for the case that all resulting Fock spaces are finite-dimensional (which requires $R,S$ to have ``Fermionic'' behaviour), such an isomorphism was recently proven by Pennig \cite{Ulrich1} with quite different methods as part of his classification of polynomial exponential functors on the category of finite-dimensional Hilbert spaces. As required for applications in quantum field theory, our result holds for infinite-dimensional Fock spaces and does not only give an isomorphism of Hilbert spaces, but also an isomorphism of representations of Zamolodchikov algebras.

The isomorphism \eqref{1} is of particular interest because up to a natural notion of equivalence, every involutive R-matrix can be written as an iterated box sum of finitely many very simple $(\pm1)$ R-matrices \cite{Y-B}. These correspond to free field theories, and the resulting decomposition of the Fock space and Zamolodchikov representations could turn out to be a tool to investigate asymptotic freedom of scaling limits. As a first step in this direction, we investigate in Section~5 under which conditions we may identify our Fock representation with tensor products of simpler ones, providing examples and counterexamples.

The article ends with an outlook to applications in quantum field theory in Section 6. We recall the definition of a regular (unitary, Hermitian analytic, crossing-symmetric, regular) two-body S-matrix and investigate in several examples when its limits as $\theta\to\pm\infty$ are involutive. 

A more detailed analysis of the short distance scaling limits will appear in a future work.

\section{An Abstract Zamolodchikov-Faddeev Algebra}
We begin by abstractly defining a variation of the well-known Zamolodchikov-Faddeev algebra \cite{ZamFadd, Faddeev:1980zy}. Let $\mathcal{L}$ be a (separable) Hilbert space and $S$ be a set of $d^4$ ($d\in \mathbb{N}$) complex numbers whose elements are labelled by symbols $S^{\alpha \beta}_{\delta \gamma}$ where $\alpha, \beta,\delta,\gamma \in \{1, \ldots, d\}$. We then define the unital $*$-algebra $\mathcal{Z}(S, \mathcal{L})$ as the algebra generated by the symbols $1_{\mathcal{Z}(S, \mathcal{L})}, Z_{1}(f), Z_2(f), \ldots, Z_d(f)$ for all $f\in \mathcal{L}$ which obey the following exchange relations: 
\begin{equation}\label{InfRels1}
Z_{\alpha}(f)Z_{\beta}(g) = S^{\beta \alpha}_{\delta \gamma} Z_{\gamma}(g)Z_{\delta}(f), 
\end{equation}
\begin{equation}\label{InfRels2}
Z_{\alpha}(f)Z^*_{\beta}(g) = S^{\alpha \gamma}_{\beta \delta}Z^*_{\gamma}(g)Z_{\delta}(f) + \delta^{\alpha}_{\beta}\cdot \langle f,g\rangle_{\mathcal{L}}1_{\mathcal{Z}(S,\mathcal{L})},
\end{equation}
where we understand that the repeated indices in an expression imply the sum over all possible values (Einstein summation convention).
\begin{remark}
We may view $S$ as a linear map over the tensor square of a Hilbert space $\mathcal{H}$ with dimension $d$. Then the numbers $S^{\alpha \beta}_{\delta \gamma}$ can be viewed as the matrix elements $\langle e_{\alpha} \otimes e_{\beta}, S(e_{\delta} \otimes e_{\gamma})\rangle $ if $(e_{\alpha})_{\alpha=1}^{d_{\mathcal{H}}}$ is an orthonormal basis of $\mathcal{H}$. Zamolodchikov algebras play a prominent role in integrable models of quantum field theory, see for example \cite{5}. A related class of algebras, the so-called \textit{Wick algebras}, are defined by omitting \eqref{InfRels2} from the definition. Their representations have been studied in \cite{ZamAlg1, ZamAlg2, ZamAlg3}.
\end{remark}
At this stage it is not at all clear whether or not there exist Hilbert space representations of $\mathcal{Z}(S, \mathcal{L})$ - for example, in \cite[p. 18]{ZamAlg2} it is shown that for certain choices of $S$ the corresponding Wick algebra admits no Hilbert space representations. As we will see, under certain assumptions on $S$ a GNS representation of $\mathcal{Z}(S,\mathcal{L})$ can be constructed.
\par
The notion of Wick ordering (sometimes also known as ``normal" ordering) is a prolific and useful concept in the analysis of these algebras. To write an element $X\in \mathcal{Z}(S, \mathcal{L})$ in Wick ordered form means to apply the governing relation \eqref{InfRels2} such that $X$ becomes of the form
\begin{equation}
\label{WickForm}
\sum_{\boldsymbol{\eta}, \boldsymbol{\xi}} \zeta_{\boldsymbol{\eta}, \boldsymbol{\xi}} Z^*_{\boldsymbol{\eta}}(f_{\boldsymbol{\eta}})Z_{\boldsymbol{\xi}}(g_{\boldsymbol{\xi}})
\end{equation}
where $\zeta_{\boldsymbol{\eta}, \boldsymbol{\xi}} \in \mathbb{C}$. The multi-index notation we have adopted here can be read as, for example,
$$Z^*_{\boldsymbol{\eta}}(f_{\boldsymbol{\eta}}) = Z^*_{\eta_{1}}(f_{\eta_{1}})Z^*_{\eta_{2}}(f_{\eta_{2}}) \cdots Z^*_{\eta_{N}}(f_{\eta_{N}}),$$
where all $f_{\eta_{n}} \in \mathcal{L}$ and $|\boldsymbol{\eta}| = N \in \mathbb{N}_0$ - in the case where $|\boldsymbol{\eta}| =0$, we take $Z^*_{\boldsymbol{\eta}}(f_{\boldsymbol{\eta}}) = 1_{\mathcal{Z}(S, \mathcal{L})}$. We also remark that for the case of $|\boldsymbol{\eta}| = |\boldsymbol{\xi}| = 0$ we have just a multiple of the identity. Every element of $\mathcal{Z}(S,\mathcal{L})$ can be written in Wick ordered form. The Wick ordered form is typically not unique as we may exchange any two $Z$ or $Z^*$ elements in the expression using \eqref{InfRels1}. However, the term with $|\boldsymbol{\eta}| = |\boldsymbol{\xi}| =0$ \textit{is} unique.

\par
This discussion facilitates the definition of a linear functional over $\mathcal{Z}(S, \mathcal{L})$ and in particular proves it to be uniquely defined by the properties we outline.
\begin{definition}\label{FinFunDef}
We define a normalised linear functional $\omega : \mathcal{Z}(S, \mathcal{L}) \to \mathbb{C}$ by the properties
\begin{enumerate}[i)]
\item \begin{equation}
\label{LinFun1}
\omega(1_{ \mathcal{Z}(S, \mathcal{L})}) = 1,
\end{equation}
\item \begin{equation}
\label{LinFun2}
\omega(Z^*_{\alpha}(f)\cdot X) =0,
\end{equation}
\item \begin{equation}
\label{LinFun3}
\omega(X\cdot Z_{\alpha}(f)) =0, 
\end{equation}
\end{enumerate}
for all $\alpha \in \{1,\ldots, d\}$, $f\in \mathcal{L}$ any $X \in  \mathcal{Z}(S, \mathcal{L})$.
\end{definition}
\par
Defining a second functional as $\lambda(X) := \overline{\omega(X^*)}$ and applying uniqueness, we see that $\omega$ is Hermitian, but it is not necessarily positive.
\begin{examples}\label{AlgExs}
We consider here some simple examples of $\mathcal{Z}(S,\mathcal{L})$.\par
If we take first $S^{ \beta \alpha}_{\delta \gamma} = \pm \delta^{\alpha}_{\delta}\delta^{\beta}_{\gamma}$, where $\delta$ is the Kronecker delta, the relations \eqref{InfRels1} and \eqref{InfRels2} now read (for $f,g \in \mathcal{L}$)
\begin{equation}\label{FinZam1}
Z_{\alpha}(f)Z_{\beta}(g) = \pm Z_{\beta}(g)Z_{\alpha}(f),
\end{equation}
\begin{equation}
\label{FinZam2}
Z_{\alpha}(f)Z^*_{\beta}(g) = \pm Z^*_{\beta}(g)Z_{\alpha}(f) + \delta^{\alpha}_{\beta}\cdot \langle f,g\rangle_{\mathcal{L}}
\end{equation}
Choosing an orthonormal basis $(e_{\alpha})_{\alpha=1}^d$ of $\mathbb{C}^d$, one realises 
$$Z_{\alpha}(f) =: a(e_{\alpha}\otimes f)$$
satisfy the governing relations of the CCR$(\mathbb{C}^d \otimes \mathcal{L})$ ($+$) and CAR$(\mathbb{C}^d \otimes \mathcal{L})$ ($-$) algebras \cite{BratRob2, EvansQA}, respectively.  We will note more on their representations in the next section.
If we instead take $S^{\alpha \beta}_{\delta \gamma} = - \delta^{\alpha}_{\delta} \delta^{\beta}_{\gamma}$, the governing relations become
\begin{equation}\label{Deg1}
Z_{\alpha}(f)Z_{\beta}(g) =  - Z_{\alpha}(g)Z_{\beta}(f)
\end{equation}
\begin{equation}\label{Deg2}
Z_{\alpha}(f)Z^*_{\beta}(g) =\delta^{\alpha}_{\beta}\left( - \sum_{\delta} Z^*_{\delta}(g)Z_{\delta}(f) +  1_{\mathcal{Z}(1, \mathcal{L})}\right).
\end{equation}
The interest in this example comes in the Fock representation (which we will discuss more in the next section), but for now we simply note that this example is also explored for the case of $\mathcal{L} = \mathbb{C}$ in \cite[p. 48]{ZamAlg2} in a Wick algebraic setting where it is known as a ``degenerate case".\par
\end{examples}
\section{Fock Representations}\par
We now consider (pre-)Hilbert space representations of $\mathcal{Z}(S,\mathcal{L})$. The motivation for considering such representations are many - for example, in a quantum field theoretic setting, they provide the framework to describe systems of particles and the observables present in a vacuum representation as (typically unbounded) operators. \par
To this end we take the tensor product $\tilde{\mathcal{H}} := \mathcal{H} \otimes \mathcal{L}$ (we reserve the notation of the tilde signifying the tensor product with $\mathcal{L}$) of a Hilbert space $\mathcal{H}$ (of finite dimension $d_{\mathcal{H}}$) and the second internal space $\mathcal{L}$ for which we do not specify (or require) finite dimensionality. The set of complex numbers $S$ we now view as an endomorphism on $\mathcal{H}^{\otimes 2}$, which we may realise as a $d_{\mathcal{H}}^2 \times d_{\mathcal{H}}^2$ matrix where the matrix elements are as mentioned previously $S^{\alpha \beta}_{\delta \gamma} = \langle e_{\alpha}\otimes e_{\beta}, S(e_{\delta}\otimes e_{\gamma})\rangle$. \par
In applications in quantum field theory, a parameter-dependent version $\mathbb{R} \ni \theta \mapsto \boldsymbol{S}(\theta)$ plays the role of a two-particle scattering operator. Such operators have a number of analytic and algebraic properties including unitarity and Hermitian analyticity. Moreover, they are solutions of the parameter-dependent Yang-Baxter equation (see \cite{AbdallaAbdallaRothe:2001} for a more detailed account of these properties, and Definition %%S-matrix definition)
The parameter-independent $S$ considered here then corresponds to $S=\boldsymbol{S}(0)$ or one of the limits $S=\lim_{\theta \to \pm \infty}(\boldsymbol{S}(\theta))$. These matrices are still far from being arbitrary, so we will therefore restrict from on to a particular class of operators $S$ in the definition of $\mathcal{Z}(S, \mathcal{L})$.
\begin{definition}
Let $\mathcal{H}$ be a Hilbert space. An \em{involutive R-matrix} on $\mathcal{H}$ is a unitary, involutive map $S\in \mathcal{B}(\mathcal{H}\otimes \mathcal{H})$ that solves the Yang-Baxter equation. That is
$$S=S^*=S^{-1},$$
\begin{equation}
\label{Y-BC}
\left(S\otimes 1_{\mathcal{H}}\right)\left(1_{\mathcal{H}} \otimes S \right)\left(S\otimes 1_{\mathcal{H}}\right) = \left(1_{\mathcal{H}} \otimes S\right)\left(S\otimes 1_{\mathcal{H}}\right)\left(1_{\mathcal{H}} \otimes S\right).
\end{equation}
We also denote by $\mathcal{R}_0(\mathcal{H})$ the set of all involutive $R$-matrices on $\mathcal{H}$.
\end{definition}
We will mostly be interested in the case where $\mathcal{H}$ is finite dimensional (which we will always explicitly state), otherwise we always take $\mathcal{H}$ to be separable.
\par
Before beginning the discussion of representations of $\mathcal{Z}(S,\mathcal{L})$ it is necessary to extend the definition of $S$ to involve the space $\mathcal{L}$ and to do so we  introduce the unitary operator $U_n : \left(\mathcal{H} \otimes \mathcal{L}\right)^{\otimes n} \to \mathcal{H}^{\otimes n} \otimes \mathcal{L}^{\otimes n}$ defined by
\begin{equation}\label{DisUni}
U_n\left( \bigotimes_{i=1}^n (h_i \otimes f_i)\right) = \left(\bigotimes_{i=1}^n h_i\right) \otimes \left(\bigotimes_{i=1}^n f_i \right), \quad h_i \in \mathcal{H}, f_i \in \mathcal{L}.
\end{equation}
\par
We can easily see that $U_n$ is unitary from the above expression and it can be thought of as ``disentangling" contributions from both Hilbert spaces, which induces an isomorphism between the domain and codomain of $U_n$, hence we will explicitly describe data acting only on $\mathcal{H}^{\otimes n} \otimes \mathcal{L}^{\otimes n}$ in this section. Employing the bounded linear operator $\mathcal{B}(\mathcal{L}\otimes \mathcal{L}) \ni F:= F_{\mathcal{L}}$ (the tensor flip) we write $$S_F := S\otimes F : \mathcal{H}^{\otimes 2} \otimes \mathcal{L}^{\otimes 2} \to  \mathcal{H}^{\otimes 2} \otimes \mathcal{L}^{\otimes 2},$$ which one readily checks is still unitary, involutive and a solution to the Yang-Baxter equation. Though this is the explicit operator used in the construction, the interest is mostly in the contributions from $S$, and so we will avoid using $S_F$ in further notation where possible.
\par
We would like to consider the GNS representation of $\mathcal{Z}(S, \mathcal{L})$ with respect to the functional $\omega$, but at this stage it is unclear if $\omega$ is positive. Instead, we will independently construct a representation of $\mathcal{Z}(S, \mathcal{L})$ and show it has the GNS property.\par
We choose an orthonormal basis $(e_{\alpha})_{\alpha =1}^{d_{\mathcal{H}}}$ of $\mathcal{H}$ and also make use of multi-index notation where we take $e_{\boldsymbol{\alpha}}\in \mathcal{H}^{\otimes n}$ to mean $e_{\alpha_1} \otimes \cdots \otimes e_{\alpha_n}$.
\par
Let $S\in \mathcal{R}_0(\mathcal{H})$ then we recall the structure of a Hilbert space representation of $\mathcal{Z}(S, \mathcal{L})$ as laid in a field theoretic setting in \cite{5}.
Denote by $S_{k,n} := 1^{\otimes (k-1)} \otimes S\otimes 1^{\otimes (n-k-1)}$ (where $S(e_{\alpha}\otimes e_{\beta}) = S^{\gamma \delta}_{\alpha \beta} e_{\gamma}\otimes e_{\delta}$) then we construct unitary operators on $\mathcal{H}^{\otimes n}\otimes \mathcal{L}^{\otimes n}$:
\begin{equation}\label{PermOps}
D_n^{S}(\tau_i) = S_{i, n}\otimes F_{i}
\end{equation}
where $\tau_i \in \mathfrak{S}_n$ (the symmetric group of $n$ letters) is a transposition, swapping nearest neighbour $i$ and $(i+1)$-th elements.
It is straightforward to see that these operators generate a unitary representation of $\mathfrak{S}_n$ on $\mathcal{H}^{\otimes n}\otimes \mathcal{L}^{\otimes n}$, then we can define an orthogonal projection \cite{5} by taking their mean:
\begin{equation}\label{InfProj}
P_n^{S} : = \frac{1}{n!} \sum_{\pi \in \mathfrak{S}_n} D_n^S(\pi) \in \mathcal{B}(\mathcal{H}^{\otimes n}\otimes \mathcal{L}^{\otimes n}).
\end{equation}
Define now the spaces
$$\tilde{\mathcal{H}}_n := U^*_nP_n^SU_n\tilde{\mathcal{H}}^{\otimes n}$$
then the \textit{$S$-symmetrised Fock space} is given by
$$\mathcal{F}_S(\tilde{\mathcal{H}}) := \bigoplus_{n\geq 0} \tilde{\mathcal{H}}_n.$$
\par
For \eqref{PermOps} to be a unitary representation of $\mathfrak{S}_n$ involutivity of $S$ is a crucial property. Dropping involutivity, unitary $R$-matrices only give representations of the Braid groups. The concept of an $S$-symmetric Fock space can be generalised to non-involutive $S$ \cite{Braided}, but it won't play a role in the current work.\par
On this space we have a vacuum vector $\Omega_S = 1\oplus 0 \oplus ...$, and a dense subspace $\mathcal{F}^0_{S}(\tilde{\mathcal{H}})$ (consisting of vectors of ``finite particle" number, meaning they are terminating direct sums of elements in increasing tensor powers). There is a natural unitary $U$ from $\mathcal{F}_S(\tilde{\mathcal{H}})$ to the ``disentangled Fock space" $\bigoplus_n P_n^S(\mathcal{H}^{\otimes n}\otimes \mathcal{L}^{\otimes n})=: \bigoplus_n \mathcal{D}_n(\tilde{\mathcal{H}})$, namely $U=\bigoplus_n U_n$ is the second quantisation of the unitaries $U_n$ \eqref{DisUni}. We may therefore treat operators on $\mathcal{F}_S(\tilde{\mathcal{H}})$ and $U\mathcal{F}_S(\tilde{\mathcal{H}})$ on the same footing. To discuss the Fock representation of $\mathcal{Z}(S, \mathcal{L})$, it is more convenient to work on the latter space, and we define
\begin{subequations}\label{eq:1}
     \begin{align}
z_S^*(e_{\xi} \otimes g)v_n \otimes f_n &:= \sqrt{n+1}P_{n+1}^S\left(e_{\xi}\otimes v_n \otimes g\otimes f_n\right)\label{CreAnnInf1}, \\
z_S(e_{\xi}\otimes g) &:= \left(z_S^*(e_{\xi}\otimes g)\right)^*\label{CreAnnInf2},
\end{align}
\end{subequations}
for $v_n\in \mathcal{H}^{\otimes n}, f_n \in \mathcal{L}^{\otimes n}, g\in \mathcal{L}.$ \par
\begin{remark}
We can write the explicit action of $z_S$ in terms of the scalar product on $\tilde{\mathcal{H}}_n$ by
\begin{equation*}
\begin{split}
\langle w_{n-1} \otimes h_{n-1}, z_S(e_{\xi}\otimes g)v_{n}\otimes f_{n} \rangle  &=\sqrt{n} \langle  e_{\xi} \otimes w_{n-1} \otimes g\otimes h_{n-1}, v_n \otimes f_n \rangle, \\ z_S(e_{\xi}\otimes g)\Omega_S &= 0,
\end{split}
\end{equation*}
for  $w_{n-1}\otimes h_{n-1} \in \mathcal{H}^{\otimes (n-1)}\otimes \mathcal{L}^{\otimes (n-1)}$. These operators then restrict to the symmetrised spaces $\mathcal{D}_n(\tilde{\mathcal{H}})$.
In bra-ket notation, \eqref{CreAnnInf1} simply reads
\begin{equation}\label{Bracket}
z_S(e_{\xi}\otimes g)v_{n}\otimes f_{n} = \sqrt{n} \langle e_{\xi}\otimes g | v_n\otimes f_n.
\end{equation}
\end{remark}
\par
We have defined $z^*_S, z_S$ with basis vectors of $\mathcal{H}$ as arguments, however we can extend the definition to operators $z_S^*(\psi)$, $\psi \in \tilde{\mathcal{H}}$ by linearity in their arguments (care to be taken when doing the same to $z_S$ as it is anti-linear in its argument).\par For ease of notation we will use the shorthand $z_S(e_{\xi} \otimes g) = z_{S,\alpha}(g)$, and the polynomial algebra generated by all $z_{S, \alpha}(f), z_{S, \beta}^*(g), 1_{\tilde{\mathcal{H}}}$ we will denote by $\mathcal{P}_S$.\par
\begin{proposition}\label{Cyclicity}
The vacuum vector $\Omega_S$ is cyclic for the algebra $\mathcal{P}_S$, that is $\mathcal{P}_S\Omega_S \subset U \mathcal{F}_S(\tilde{\mathcal{H}})$ is dense.
\end{proposition}
\begin{proof}
Let $\psi \in \mathcal{F}_S(\tilde{\mathcal{H}})$ such that $\psi$ is orthogonal to $\mathcal{P}_S\Omega$. For any $n \in \mathbb{N}_0$, and vectors $f_1, \ldots, f_n \in \mathcal{L}$, $\alpha_1,\ldots, \alpha_n \in \{1, \ldots, d_{\mathcal{H}}\}$ we then have
\begin{equation*}
\begin{split}
0 &= \langle \psi, z_{S,\alpha_1}^*(f_1) \cdots z_{S, \alpha_n}^*(f_n) \Omega_S \rangle \\
&= \sqrt{n!}\langle \psi, P_n^S(e_{\alpha_1} \otimes \cdots \otimes e_{\alpha_n} \otimes f_1 \otimes \cdots \otimes f_n)\rangle \\ 
&= \sqrt{n!}\langle \psi, e_{\alpha_1} \otimes \cdots \otimes e_{\alpha_n} \otimes f_1\otimes \cdots \otimes f_n \rangle
\end{split}
\end{equation*}
where we have used that the projection $P_n^S$ is self-adjoint and leaves the symmetrised vector $\psi$ invariant. Since vectors of the form $e_{\alpha_1} \otimes \cdots \otimes e_{\alpha_n} \otimes f_1 \otimes \cdots \otimes f_n$ form a total set in $\mathcal{H}^{\otimes n}\otimes \mathcal{L}^{\otimes n}$, we conclude that $\psi =0$. Thus $\Omega_S$ is cyclic for $\mathcal{P}_S$.
\end{proof}
Before moving to the next result we note specific elements of $\mathfrak{S}_n$ as they will play an important role in the following proof - define $\sigma_n := \tau_{n-1} \tau_{n-2} ... \tau_1 \in \mathfrak{S}_n$ which acts by taking the first element and moving it to the $n$-th position.
\begin{theorem}\label{theorem:GNS}
Let $\mathcal{H}$ be a finite dimensional Hilbert space and $S \in \mathcal{R}_0(\mathcal{H})$. Then the map $\pi_S: \mathcal{Z}(S, \mathcal{L}) \to \mathcal{P}_S$ given by
\begin{equation}
\pi_S(1_{\mathcal{Z}(S, \mathcal{L})}) := 1_{\tilde{\mathcal{H}}}, \quad \pi_S(Z_{\alpha}(f)) := z_{S, \alpha}(f)
\end{equation}
extends to a unital $*$-representation of $\mathcal{Z}(S, \mathcal{L})$ on $\mathcal{F}^0_S(\tilde{\mathcal{H}})$ with cyclic vector $\Omega_S$ and
\begin{equation}\label{Functional}
\omega(X) = \langle \Omega_S, \pi_S(X)\Omega_S \rangle. \qquad (X\in \mathcal{Z}(S, \mathcal{L}))
\end{equation}
\end{theorem}
\begin{proof}
We show first that the operators $z_{S, \alpha}(f), z^*_{S,\alpha}(f)$ satisfy \eqref{InfRels1} and \eqref{InfRels2} for all $f \in \mathcal{L}, \alpha \in \{1,\ldots, d_{\mathcal{H}}\}$. Let $v_n \otimes h_n \in \mathcal{D}_n(\tilde{\mathcal{H}})$, $f, g\in \mathcal{L}$ then taking into account $P_{n+2}^S = P_{n+2}^S S_{1, n+1} \otimes F_1$
\begin{equation*}
\begin{split}
z^*_{S,\alpha}(f)z^*_{S,\beta}(g) v_n \otimes h_n &= \sqrt{n+1}\sqrt{n+2} P_{n+2}^S \left(e_{\alpha} \otimes e_{\beta} \otimes v_n \otimes f\otimes g\otimes h_n\right) \\
&=  \sqrt{n+1}\sqrt{n+2} P_{n+2}^S \left( S^{\gamma \delta}_{\alpha \beta} e_{\gamma}\otimes e_{\delta} \otimes v_n \otimes g\otimes f\otimes h_n\right) \\
&= S^{\gamma \delta}_{\alpha \beta} z^*_{S,\gamma}(g)z^*_{S,\delta}(f) v_n \otimes h_n.
\end{split}
\end{equation*}
Given that $n$ and $v_n\otimes h_n$ were arbitrary, we read off
$$z^*_{S, \alpha}(f)z^*_{S, \beta}(g) = S^{\gamma \delta}_{\alpha \beta} z^*_{S,\gamma}(g) z^*_{S,\delta}(f)$$
as operators on $\mathcal{F}^0_S(\tilde{\mathcal{H}})$. Taking adjoints of both sides and applying both the unitarity and involutivity of $S$, we arrive at equation \eqref{InfRels1}. \par
For showing \eqref{InfRels2}, we first compute the action of the first term right hand side on some $v_n\otimes h_n \in \mathcal{D}_n(\tilde{\mathcal{H}})$. It is enough to do so on  vectors of the form $v_n = e_i\otimes v_{n-1}$, $h_n = a \otimes h_{n-1}$ $(a\in \mathcal{L}, v_{n-1}\otimes h_{n-1} \in \mathcal{D}_{n-1}(\tilde{\mathcal{H}}))$ as they form a total set in $\mathcal{H}^{\otimes n}$ and $\mathcal{L}^{\otimes n}$, respectively. We have:
\begin{equation*}
\begin{split}
S^{\alpha \gamma}_{\beta \delta} z^*_{S, \gamma}(g)z_{S, \delta}(f) v_n\otimes h_n &= 
\sqrt{n} S^{\alpha \gamma}_{\beta \delta} z^*_{S, \gamma}(g) \langle e_{\delta} \otimes f| (e_i \otimes v_{n-1} \otimes a\otimes h_{n-1}) \\
&=n S^{\alpha \gamma}_{\beta \delta} \delta^{\delta}_i \langle f, a\rangle P^S_{n}  \left(e_{\gamma} \otimes v_{n-1} \otimes g \otimes h_{n-1}\right)\\
&= n S^{\alpha \gamma}_{\beta i} \langle f, a\rangle  P^S_{n} \left(e_{\gamma} \otimes v_{n-1} \otimes g \otimes h_{n-1}\right).
\end{split}
\end{equation*}
Since $v_{n-1}$ and $h_{n-1}$ are correctly symmetrised, the action of the projection $P_n^S$ in the final line simplifies. Namely, we need only sum over the permutations that shift the $e_{\gamma}$ and $g$ terms such that they appear in each tensor slot \cite{Gand1}; $P_n^S = \frac{1}{n} \sum_{k=1}^n D_n^S(\sigma_k)\cdot 1\otimes P_{n-1}^S$. The above now reads
\begin{equation}\label{Rightside}
S^{\alpha \gamma}_{\beta \delta} z^*_{S, \gamma}(g)z_{S, \delta}(f) v_n\otimes h_n = S^{\alpha \gamma}_{\beta i} \langle f,a\rangle \sum_{k=1}^n D_n^S(\sigma_k) \left(e_{\gamma} \otimes v_{n-1} \otimes g\otimes h_{n-1}\right).
\end{equation}
To compute the left hand side of \eqref{InfRels2}, we now consider its scalar product with a $w_n\otimes b_n := e_j \otimes w_{n-1} \otimes c\otimes b_{n-1} \in \mathcal{D}_n(\tilde{\mathcal{H}})$ in the scalar product:
\begin{equation*}
\begin{split}
&\langle w_n \otimes b_n, z_{S,\alpha}(f) z^*_{S,\beta}(g) v_n\otimes h_n \rangle = \langle z^*_{S,\alpha}(f)w_n \otimes b_n, z^*_{S,\beta}(g) v_n\otimes h_n \rangle  \\
&= (n+1) \langle e_{\alpha}\otimes w_n \otimes f\otimes b_n, P_{n+1}^S \left(e_{\beta}\otimes v_n \otimes g\otimes h_n\right) \rangle.
\end{split}
\end{equation*}
Since $v_n$ and $h_n$ are correctly symmetrised, the projection $P_{n+1}^S$ again simplifies as before. Noting further that $\sigma_1 = 1,$ this gives:
\begin{equation*}
\begin{split}
&\langle w_n \otimes b_n, z_{S,\alpha}(f) z^*_{S,\beta}(g) v_n\otimes h_n \rangle\\
& = \left\langle e_{\alpha}\otimes w_n \otimes f\otimes b_n, \sum_{k=1}^{n+1} D_{n+1}^S(\sigma_k) \left(e_{\beta}\otimes v_n \otimes g\otimes h_n\right) \right\rangle\\
&= \left\langle e_{\alpha}\otimes w_n \otimes f\otimes b_n, \left[1 +  \sum_{k=2}^{n+1} D_{n+1}^S(\sigma_k)\right] \left(e_{\beta}\otimes v_n \otimes g\otimes h_n\right) \right\rangle \\
&=\delta^{\alpha}_{\beta}\langle f,g\rangle \langle w_n\otimes b_n, v_n\otimes h_n\rangle \\
&\qquad \qquad + \left\langle e_{\alpha}\otimes w_n \otimes f\otimes b_n, \sum_{k=2}^{n+1} D_{n+1}^S(\sigma_k) \left(e_{\beta}\otimes e_i\otimes  v_{n-1} \otimes g\otimes a \otimes  h_{n-1}\right) \right\rangle.
\end{split}
\end{equation*}
To shift the index of the sum $\sum_{k=2}^{n+1} D_{n+1}^S(\sigma_k)$ we note that it sums over the permutations shifting the $e_{\beta}$ term through each tensor slot with the first term being the permutation given by just $S_{1,n+1} \otimes F_{1} = D_{n+1}^S(\tau_1)$. If we extract this term from the sum, we can read the remaining terms as taking the second tensor slot and permuting through the remaining $n$ slots with the first slot being untouched \cite{Gand1}. Concretely, this means that we can write this as 
$$\sum_{k=2}^{n+1} D_{n+1}^S(\sigma_k) = \left(\sum_{k=1}^n 1\otimes D_n^S(\sigma_k)\right)D_{n+1}^S(\tau_1),$$
which we use and continue in the calculation:
\begin{equation*}
\begin{split}
&\langle w_n \otimes b_n, z_{S,\alpha}(f) z^*_{S,\beta}(g) v_n\otimes h_n \rangle =\delta^{\alpha}_{\beta}\langle f,g\rangle \langle w_n\otimes b_n, v_n\otimes h_n\rangle \\
& + S^{\xi \gamma}_{\beta i} \left\langle e_{\alpha}\otimes e_j \otimes  w_{n-1} \otimes f\otimes b_n, \sum_{k=1}^{n} 1\otimes D_n^S(\sigma_k) \left(e_{\xi}\otimes e_{\gamma} \otimes v_{n-1} \otimes a\otimes g\otimes h_{n-1}\right) \right\rangle \\
&= \delta^{\alpha}_{\beta}\langle f,g\rangle \langle w_n\otimes b_n, v_n\otimes h_n\rangle \\
& \qquad \qquad + S^{\alpha \gamma}_{\beta i} \langle f,a\rangle \left\langle w_{n} \otimes  b_n, \sum_{k=1}^{n}D_n^S(\sigma_k) \left( e_{\gamma} \otimes v_{n-1} \otimes g\otimes h_{n-1}\right) \right\rangle.
\end{split}
\end{equation*}
Since all elements involved were arbitrary, we now read off what we have computed by comparing the right hand slot of the scalar product:
\begin{equation}\label{Lefthand}
z_{S,\alpha}(f) z^*_{S,\beta}(g) v_n\otimes h_n =\delta^{\alpha}_{\beta}\langle f,g\rangle v_n\otimes h_n +  S^{\alpha \gamma}_{\beta i} \langle f,a\rangle \sum_{k=1}^{n}D_n^S(\sigma_k) \left( e_{\gamma} \otimes v_{n-1} \otimes g\otimes h_{n-1}\right).
\end{equation}
We can now read that (up to the contraction term) we have equality between \eqref{Rightside} and \eqref{Lefthand}, and thus \eqref{InfRels2} is satisfied so $\pi_S$ is indeed a representation of $\mathcal{Z}(S,\mathcal{L})$ on $U\mathcal{F}_S^0(\tilde{\mathcal{H}})$. \par
Moreover, $\Omega_S$ is cyclic for this representation by Proposition \ref{Cyclicity} and then the GNS property follows once we realise that the functional defined by the right hand side of \eqref{Functional} satisfies the properties of $\omega$ as outlined in Definition \ref{FinFunDef} since $\Omega_S$ is a normalised vector and $z_S$ annihilates it.
\end{proof}
\begin{remark}
It is now apparent that $\omega$ is positive (for $S \in \mathcal{R}_0(\mathcal{H})$). For any $X\in \mathcal{Z}(S, \mathcal{L})$:
\begin{equation*}
\begin{split}
\omega(X^*X) &= \langle \Omega_S, \pi_S(X^*X) \Omega_S\rangle\\
& = \langle \Omega_S, \pi_S(X)^*\pi_S(X)\Omega_S\rangle \\
&= \langle \pi_S(X)\Omega_S, \pi_S(X)\Omega_S \rangle \\
&= \| \pi_S(X)\Omega_S \|^2 \geq 0.
\end{split}
\end{equation*}
\end{remark}
Revisiting the examples outlined in the previous section, we see that for $S^{\alpha \beta}_{\delta \gamma} = \pm \delta^{\alpha}_{\gamma}\delta^{\beta}_{\delta}$ we arrive at the totally symmetric ($+$) or totally antisymmetric ($-$) Fock space over $\tilde{\mathcal{H}}$, usually known as the Bosonic and Fermionic Fock spaces, respectively.\par
The case of $S^{\alpha \beta}_{\delta \gamma} = -\delta^{\alpha}_{\delta}\delta^{\beta}_{\gamma}$ ($S=-1$) results in a very small space for $\mathcal{L} = \mathbb{C}$: For this particular choice of $S$, the projection simplifies greatly to 
$$P_n^S = \sum_{\pi \in \mathfrak{S}_n} \text{sgn}(\pi) = \begin{cases} 1,&\ n=1 \\ 0,& \ n>1 \end{cases}$$
where sgn$(\pi)$ is the sign of the permutation $\pi$. There then only exists a zero-particle space (multiples of the vacuum) and the single particle space $\tilde{\mathcal{H}}$:
$$\mathcal{F}_S(\tilde{\mathcal{H}}) = \Omega_S \oplus \tilde{\mathcal{H}}.$$
The other extreme is given by $S = 1$ and $\mathcal{L} = \mathbb{C}$, which generates the full unsymmetrised (or ``Boltzmann") Fock space
$$\mathcal{F}_S(\mathcal{H}) = \bigoplus_{n\geq 0} \mathcal{H}^{\otimes n}.$$
In general, $\mathcal{F}_S(\mathcal{H})$ ``interpolates" between these two extreme cases.
\section{Operations on R-Matrices and Equivalences}
In this work we aim to generalise the following notion (sometimes referred to as an ``exponential relation", see for example \cite{BSZ} and references therein): Let $\mathcal{H}, \mathcal{K}$ be Hilbert spaces, and and denote by $\mathcal{F}_{\pm}(\mathcal{H})$ the Bosonic/Fermionic Fock space over $\mathcal{H}$. To keep touch with our earlier constructions, this corresponds to the Fock representation of the algebra $\mathcal{Z}(\pm F_{\mathcal{H}}, \mathbb{C})$, where $F_{\mathcal{H}}$ is the tensor flip on $\mathcal{H}^{\otimes 2}$ (we simplify notation here using $\pm$ as a subscript to mean $S = \pm F_{\mathcal{H}}$ to remain familiar with notation existing in the literature). Then it is well known that there exists a natural isomorphism
\begin{equation}\label{ExpRel}
\mathcal{F}_{\pm}(\mathcal{H}\oplus \mathcal{K}) \cong \mathcal{F}_{\pm}(\mathcal{H}) \otimes \mathcal{F}_{\pm}(\mathcal{K}).
\end{equation}\par
This isomorphism does not only hold as a Hilbert space isomorphism (which is trivial for the case of infinite dimensional separable Hilbert spaces), but also as an isomorphism of representations of the CCR/CAR algebras.\par
In the following, we will explore to which extent \eqref{ExpRel} generalises to our setting of $S$-symmetric Fock spaces. As a prerequisite for doing so we need to compare $R$-matrices on tensor products and direct sums of Hilbert spaces. The relevant notions are recalled below.
\begin{definition}\emph{\cite{Y-B}}\label{BoxPlus}
Let $\mathcal{H}, \mathcal{K}$ be Hilbert spaces, and $S\in \mathcal{R}_0(\mathcal{H}),  R\in \mathcal{R}_0(\mathcal{K})$. Then we define
\begin{enumerate}[i)]
\item $S\boxplus R : (\mathcal{H} \oplus \mathcal{K}) \otimes (\mathcal{H} \oplus \mathcal{K})  \to (\mathcal{H} \oplus \mathcal{K}) \otimes (\mathcal{H} \oplus \mathcal{K}) $ as
$$S\boxplus R  := S \oplus R \oplus F, \quad \text{on}$$
$$(\mathcal{H} \oplus \mathcal{K}) \otimes (\mathcal{H} \oplus \mathcal{K})  = (\mathcal{H} \otimes \mathcal{H} ) \oplus (\mathcal{K} \otimes \mathcal{K}) \oplus (\mathcal{H} \otimes \mathcal{K}) \oplus (\mathcal{K} \otimes \mathcal{H}).$$
\item 
$S\boxtimes R : \mathcal{H}\otimes \mathcal{K} \otimes \mathcal{H}\otimes \mathcal{K} \to \mathcal{H}\otimes \mathcal{K}\otimes \mathcal{H}\otimes \mathcal{K} $ as
$$S\boxtimes R = F_2(S\otimes R)F_2$$
where $F_2$ exchanges the second and third tensor factors.
\end{enumerate}
We use the terminology \textit{box-sum} and \textit{box-product}, respectively, for these operations.
\end{definition}
\par
These two operations preserve unitarity, involutivity and if $S, R$ are solutions to the Yang-Baxter equation, then $S\boxplus R, S\boxtimes R$ are also solutions and hence we have that $S\boxplus R \in \mathcal{R}_0(\mathcal{H}\oplus \mathcal{K})$ and $S\boxtimes R \in \mathcal{R}_0(\mathcal{H} \otimes \mathcal{K})$ \cite{Y-B}.\par
A generalisation of \eqref{ExpRel} that we will show in this work (for now setting $\mathcal{L} = \mathbb{C}$) now reads as
\begin{equation}\label{NewExpRel}
\mathcal{F}_{S\boxplus R}(\mathcal{H}\oplus \mathcal{K}) \cong \mathcal{F}_{S}(\mathcal{H}) \otimes \mathcal{F}_R(\mathcal{K}),
\end{equation}
where $S\in \mathcal{R}_0(\mathcal{H}), R\in \mathcal{R}_0(\mathcal{K})$ for Hilbert spaces $\mathcal{H}, \mathcal{K}$. Setting $S= F_{\mathcal{H}}, R= F_{\mathcal{K}}$ then $S\boxplus R = F_{\mathcal{H} \oplus \mathcal{K}}$ and the above reads the same as \eqref{ExpRel} for the Bose case, so we at least realise immediately this is consistent with existing results.\par
Under the assumption that all resulting Fock spaces are finite dimensional, the isomorphism \eqref{NewExpRel} has recently been established by Pennig using exponential functors \cite{Ulrich1}. This assumption of finite dimensionality is satisfied for example by $S=-1$ or $S=-F_{\mathcal{H}}$, and more generally for all $R$-matrices having only ``Fermionic" Thoma parameters. For applications to quantum field theory, however, it is essential to have an infinite dimensional Hilbert space, and our setup and arguments in the following will be quite different from that of \cite{Ulrich1}.
\par
In addition, the analysis of equivalent functors in \cite{Ulrich1} also yielded an equivalence between functors when there is a natural equivalence between the $R$-matrices associated to them. We may also wonder if it is possible to formulate such an isomorphism in our setting and to this end, we recall the notion of equivalent $R$-matrices.
\begin{definition}\emph{\cite{Y-B}}
Let $\mathcal{H}, \mathcal{K}$ be Hilbert spaces, $S \in \mathcal{R}_0(\mathcal{H}), R\in \mathcal{R}_0(\mathcal{K})$, then they are said to be equivalent - denoted as $S \sim R$ - if and only if, for each $n \in \mathbb{N}$ the representations $D_n^S$ and $D_n^R$ are unitarily equivalent.
\end{definition}
Let $S\in \mathcal{R}_0(\mathcal{H}), R\in \mathcal{R}_0(\mathcal{K})$. If $S \sim R$ this definition means that there exists a unitary intertwining operator $Y^{S,R}_n : \mathcal{H}_n \to \mathcal{K}_n$ such that
\begin{equation}\label{InterOps}
Y_n^{S,R}D_n^S(\pi)= D_n^R(\pi)Y_n^{S,R}. \qquad \quad (\pi \in \mathfrak{S}_n)
\end{equation}
In general, the form of $Y_n^{S,R}$ is unknown, but we can provide a few examples \cite{Y-B} of when we may write down its action explicitely.
\begin{itemize}
\item\label{type1} \textbf{Type 1:} There exists a unitary $Q$ on $\mathcal{H}$ such that $(Q\otimes Q) S(Q^*\otimes Q^*) = R$. Then $S\sim R$ and we may choose
$$Y_n^{S,R} = Q^{\otimes n}.$$
\item \textbf{Type 2:} There exists a unitary $Q$ on $\mathcal{H}$ such that $[S, Q\otimes Q] =0$ and $(1\otimes Q)S(1\otimes Q^*) = R.$ Then $S\sim R$ and we may choose
$$Y_n^{S,R} = 1\otimes Q \otimes \cdots \otimes Q^{n-1}.$$
\item \textbf{Type 3:} Let $F : \mathcal{H}^{\otimes 2}\to \mathcal{H}^{\otimes 2}$, $F(x\otimes y) = y\otimes x, (x,y \in \mathcal{H})$ be the ``flip" operator, such that $FSF = R$. Then $S\sim R$ and we may choose
$$Y_n^{S,R} = D_n^{FSF}(\iota_n)^{-1}D_n^F(\iota_n)$$
where $\iota_n$ is the total inversion permutation in $n$ letters.
\end{itemize}
\par
The significance of the equivalence relation $\sim$ stems from the fact that every $S\in \mathcal{R}_0(\mathcal{H})$ is equivalent to a very simple $R$-matrix, namely an $R$-matrix of so-called {\em normal form}: It was shown in \cite{Y-B} that for any $S\in \mathcal{R}_0(\mathcal{H})$, there exists $N\in \mathbb{N}$, dimension parameters $d_1, \ldots, d_N \in \mathbb{N}$ with $d_1+\cdots +d_N = d_{\mathcal{H}}$ and signs $\varepsilon_1, \ldots, \varepsilon_N \in \{+1, -1\}$, such that
\begin{equation}\label{Normal}
\displaystyle S \sim \bigbplus_{i=1}^N\  \varepsilon_i 1_{d_i^2}.
\end{equation}
From this definition we can read off that the tensor flip $F_{\mathcal{H}} = \bigbplus_{i=1}^{d_{\mathcal{H}}} 1$ is a normal form, and also the identity (take $N=1$ and $\varepsilon_1 = +1$), as examples. Considering now two equivalent $R$-matrices $S\sim R$, we may wonder whether we have an isomorphism
\begin{equation}\label{EquivFockIso}
\mathcal{F}_S(\mathcal{H}) \cong \mathcal{F}_R(\mathcal{K}),
\end{equation}
where this could simply be an isomorphism of Hilbert spaces, or even an isomorphism of representations of Zamolodchikov algebras.
\par
Since any $R$-matrix $S$ is equivalent to a normal form \eqref{Normal}, the combination of the anticipated isomorphisms \eqref{EquivFockIso} and \eqref{NewExpRel} would allow us to split $\mathcal{F}_S(\mathcal{H})$ into a tensor product of Fock spaces of the simple form $\mathcal{F}_{\pm 1_{d_i^2}}(\mathcal{H}_i)$.\par
As a preparatory step to the next section where we cement these ideas, we note the following results.
\begin{lemma}\label{BoxTimesReps}
Let $\mathcal{H}, \mathcal{K}$ be seperable Hilbert spaces, and $S\in \mathcal{R}_0(\mathcal{H}), R\in \mathcal{R}_0(\mathcal{K})$. Then the representation of the symmetric group, $D_n^{S\boxtimes R},$ generated by $S\boxtimes R$ is unitarily equivalent to $D_n^S\otimes D_n^R$ for any $n\in \mathbb{N}$.
\end{lemma}
\begin{proof}
We show the result for only the generating elements $\tau_k$, of $\mathfrak{S}_n$. Let $(h_{\alpha})_{\alpha \in \mathbb{N}}$ and $(k_{\beta})_{\beta \in \mathbb{N}}$ as orthonormal bases. Employing the operator $U_n$, an element in the domain of $D_n^{S\boxtimes R}$ is mapped to an element in the domain of $D_n^S\otimes D_n^R$ and the action of the latter operator is given by
\begin{equation*}
\begin{split}
&D_n^S(\tau_k)\otimes D_n^R(\tau_k)  \left(\bigotimes_{i=1}^n h_{\alpha_i}\right) \otimes \left(\bigotimes_{i=1}^n k_{\beta_i}\right) \\
& = S^{\alpha_k \alpha_{k+1}}_{\delta \gamma} R^{\beta_k \beta_{k+1}}_{\eta \xi}\left( \left(\bigotimes_{i=1}^{k-1} h_{\alpha_i} \right) \otimes h_{\delta}\otimes h_{\gamma} \otimes \left(\bigotimes_{i=k+2}^n h_{\alpha_i}\right)\right) \\
& \qquad \qquad \qquad \otimes \left( \left(\bigotimes_{i=1}^{k-1} k_{\beta_i}\right)\otimes k_{\eta} \otimes k_{\xi}\otimes \left(\bigotimes_{i=k+2}^n k_{\beta_i}\right)\right),
\end{split}
\end{equation*}
where the implicit sums converge in norm topology. Applying the linear operator $U_n^*$ to the above gives the action of $D_n^{S\boxtimes R}$ as stated.
\end{proof}
\begin{corollary}\label{Coro1}
Let $\mathcal{H}_1, \mathcal{H}_2$, $\mathcal{K}_1, \mathcal{K}_2$ be Hilbert spaces, $S_1  \in \mathcal{R}_0(\mathcal{H}_1), R_1\in \mathcal{R}_0(\mathcal{K}_1)$, $S_2  \in \mathcal{R}_0(\mathcal{H}_2), R_2 \in \mathcal{R}_0(\mathcal{K}_2)$ such that $S_1 \sim S_2$ and $R_1 \sim R_2$ . Then an intertwiner for $S_1\boxtimes R_1 \sim S_2\boxtimes R_2$, $Y^{S_1\boxtimes R_1, S_2 \boxtimes R_2}$, is given by
$$ Y^{S_1\boxtimes R_1, S_2\boxtimes R_2}_n  = U_n^* \left(Y_n^{S_1, S_2} \otimes Y_n^{R_1, R_2} \right) U_n$$
with $Y^{S_1, S_2}$, $Y^{R_1, R_2}$ intertwiners between $S_1, S_2$ and $R_1, R_2$, respectively.
\end{corollary}
\begin{proof}
This is clear from the definition of $U_n$ and Lemma \ref{BoxTimesReps}.
\end{proof}
\par
We mention as an aside that an analogue of Corollary \ref{Coro1} also holds for box-sums $S_1\boxplus R_1 \sim S_2\boxplus R_2$. As we will not need this here, we omit the details.
\section{Isomorphisms Between Polynomial Algebras and Equivalences of Representations}
We now go on to discuss generalisations of \eqref{ExpRel} and in particular, we consider \eqref{NewExpRel} with the addition of the Hilbert space $\mathcal{L}$ appearing in a tensor product with both $\mathcal{H}, \mathcal{K}$. For $S\in \mathcal{R}_0(\mathcal{H}), R\in \mathcal{R}_0(\mathcal{K})$, define $S\boxtimes F =: \tilde{S} \in \mathcal{R}_0(\tilde{\mathcal{H}})$, for $F$ the tensor flip on $\mathcal{L}\otimes \mathcal{L}$, and similarly for $R$. As mentioned previously, the tilde appearing above $R$-matrices always signifies a box-product with $F$, and above a Hilbert space always means a tensor product with the same space $\mathcal{L}$. With this notation, we will now aim to show the following:
\begin{equation}\label{MyExpRel}
\mathcal{F}_{\tilde{S}\boxplus \tilde{R}}(\tilde{\mathcal{H}}\oplus \tilde{\mathcal{K}}) \cong \mathcal{F}_{\tilde{S}}(\tilde{\mathcal{H}}) \otimes \mathcal{F}_{\tilde{R}}(\tilde{\mathcal{K}}).
\end{equation}
\par
On the left hand side, the GNS representation is already described in Section 3, where we have a space symmetrised by the operator $\tilde{S}\boxplus \tilde{R}$, but so far we have not considered representations of Zamolodchikov operators on $\mathcal{F}_{\tilde{S}}(\tilde{\mathcal{H}}) \otimes \mathcal{F}_{\tilde{R}}(\tilde{\mathcal{K}})$. We will first build data on this space - most notably the analogue of the creation/annihilation operators and vacuum vector. The exchange relations between the former and the cyclicity of the latter will be shown, before a GNS-type argument will prove that they are in fact equivalent representations of the same algebra $\mathcal{Z}(S\boxplus R, \mathcal{L})$. \par
\par
We begin with the algebra $\mathcal{Z}(S\boxplus R, \mathcal{L})$ as given in Section 1, with $S\in \mathcal{R}_0(\mathcal{H}), R\in \mathcal{R}_0(\mathcal{K})$ and note a property regarding distributivity of the box-product over the box-sum. Generally the distributivity property is only known up to equivalence, but in the specific cases considered in this work, we show that we in fact have equality.
\begin{lemma}\label{DistLem}
Let $\mathcal{H}, \mathcal{K}$ be finite dimensional Hilbert spaces, $S\in \mathcal{R}_0(\mathcal{H}), R\in \mathcal{R}_0(\mathcal{K})$. Then
\begin{equation}\label{DistProp}
\left(S\boxplus R\right)\boxtimes F = \tilde{S} \boxplus \tilde{R}.
\end{equation}
\end{lemma}
\begin{proof}
On the level of the spaces they act on, we first note that 
\begin{equation*}
\begin{split}
\mathcal{D}\left(\left(S\boxplus R\right)\boxtimes F\right)&= \left(\left(\mathcal{H}\oplus \mathcal{K}\right)\otimes \mathcal{L} \right)^{\otimes 2} \\
& \cong (\tilde{\mathcal{H}} \oplus \tilde{\mathcal{K}} )^{\otimes 2} = \mathcal{D}(\tilde{S}\boxplus \tilde{R}),
\end{split}
\end{equation*}
where by $\mathcal{D}(S)$ we mean the domain of $S$.\par
To show they do indeed map vectors in their  respective domains to the same vector, we consider each orthogonal component of their domains and discuss how each operator acts. \par
Firstly, the case of $(\mathcal{H} \otimes \mathcal{L})^{\otimes 2}$, the left hand side of \eqref{DistProp} first applies $F_2$ and then acts as a flip in the $\mathcal{L}^{\otimes 2}$ parts, and as $S\boxplus R$ on $\mathcal{H}^{\otimes 2}$ which by definition simply acts as just $S$, and finally applies a second $F_2$. More simply put, it flips the contributions from $\mathcal{L}$ and acts as $S$ on the contributions from $\mathcal{H}$. The same occurs on the right hand side of \eqref{DistProp}, where we see that it simply acts as only $\tilde{S}$ by definition of the box-sum, which flips the contributions from $\mathcal{L}$ and acts as $S$ on the contributions from $\mathcal{H}$.\par
Similarly, for the case of $(\mathcal{K} \otimes \mathcal{L})^{\otimes 2}$, the left hand side of \eqref{DistProp} flips on $\mathcal{L}$ and acts as $R$ on $\mathcal{K}$. Identically, the right hand side of \eqref{DistProp} acts just as $\tilde{R}$ which again flips on $\mathcal{L}$ and acts as $\mathcal{K}$ and so we equality again.\par
The remaining cases to consider are $\mathcal{H} \otimes \mathcal{L} \otimes \mathcal{K} \otimes \mathcal{L}$ and $\mathcal{K} \otimes \mathcal{L} \otimes \mathcal{H}\otimes \mathcal{L}$. However, since we have a single contribution from both $\mathcal{H}$ and $\mathcal{K}$ appearing, each operator simply reduces to a combination of flips acting on the appropriate spaces and it is easy to realise that $F_2(F_{\mathcal{H}\oplus \mathcal{K}} \otimes F)F_2 = F_{\tilde{\mathcal{H}}\otimes \tilde{\mathcal{K}}}$ and $F_2(F_{\mathcal{K}\oplus \mathcal{H}} \otimes F)F_2 = F_{\tilde{\mathcal{K}}\otimes \tilde{\mathcal{H}}}$. \par
Both sides of \eqref{DistProp} then act in the same way on each orthogonal part of their isomorphic domains, therefore they are equal.
\end{proof}
\par
The algebra of interest in this section, $\mathcal{Z}(S\boxplus R,\mathcal{L})$, is described by the operator $\left(S\boxplus R\right)\boxtimes F$, but now Lemma \ref{DistLem} allows us to work instead with $\tilde{S} \boxplus \tilde{R}$.\par 
We consider the Fock space $\mathcal{F}_{\tilde{S} \boxplus \tilde{R}} (\tilde{\mathcal{H}} \oplus \tilde{\mathcal{K}})$ on which we have a vacuum vector $\Omega_{\tilde{S} \boxplus \tilde{R}}$ and creation/annihilation operators $z_{\tilde{S} \boxplus \tilde{R}}^*, z_{\tilde{S} \boxplus \tilde{R}}$. The latter obey exchange relations involving the operator $\tilde{S} \boxplus \tilde{R},$ which we note here for convenience. We adopt the shorthand notation $\tilde{\mathcal{H}} \ni f_{\alpha} := e_{\alpha} \otimes f$ for $f\in \mathcal{L}$ and basis vectors $e_{\alpha}$ of $\mathcal{H}$, and $\tilde{\mathcal{K}} \ni g_{\xi} := k_{\xi} \otimes g$ for $g\in \mathcal{L}$ and basis vectors $k_{\xi}$ of $\mathcal{K}$. Then
\begin{equation}\label{IntRels1}
z_{\tilde{S} \boxplus \tilde{R}}(f_{\alpha}\oplus 0)z_{\tilde{S} \boxplus \tilde{R}}(g_{\beta}\oplus 0) = S^{\beta \alpha}_{\delta \gamma}z_{\tilde{S} \boxplus \tilde{R}}(g_{\gamma}\oplus 0)z_{\tilde{S} \boxplus \tilde{R}}(f_{\delta}\oplus 0),
\end{equation}
\begin{equation}\label{IntRels2}
z_{\tilde{S} \boxplus \tilde{R}}(f_{\alpha}\oplus 0)z^*_{\tilde{S} \boxplus \tilde{R}}(g_{\beta}\oplus 0) = S^{ \alpha \gamma}_{\beta \delta }z^*_{\tilde{S} \boxplus \tilde{R}}(g_{\gamma}\oplus 0)z_{\tilde{S} \boxplus \tilde{R}}(f_{\delta}\oplus 0) + \delta^{\alpha}_{\beta}\langle f,g\rangle \cdot 1_{\tilde{\mathcal{H}}\oplus \tilde{\mathcal{K}}},
\end{equation}
\begin{equation}\label{IntRels3}
z_{\tilde{S} \boxplus \tilde{R}}(0\oplus f_{\xi})z_{\tilde{S} \boxplus \tilde{R}}(0\oplus g_{\eta}) = R^{\eta \xi}_{\epsilon \pi}z_{\tilde{S} \boxplus \tilde{R}}(0\oplus g_{\pi})z_{\tilde{S} \boxplus \tilde{R}}(0\oplus f_{\epsilon}),
\end{equation}
\begin{equation}\label{IntRels4}
z_{\tilde{S} \boxplus \tilde{R}}(0\oplus f_{\xi})z^*_{\tilde{S} \boxplus \tilde{R}}(0\oplus g_{\eta}) = R^{\xi \pi}_{\eta \epsilon}z^*_{\tilde{S} \boxplus \tilde{R}}(0\oplus g_{\pi})z_{\tilde{S} \boxplus \tilde{R}}(0\oplus f_{\epsilon}) + \delta^{\xi}_{\eta}\langle f,g\rangle\cdot 1_{\tilde{\mathcal{H}}\oplus \tilde{\mathcal{K}}},
\end{equation}
\begin{equation}\label{IntRels5}
z_{\tilde{S} \boxplus \tilde{R}}(f_{\alpha}\oplus 0)z_{\tilde{S} \boxplus \tilde{R}}(0\oplus g_{\eta}) = z_{\tilde{S} \boxplus \tilde{R}}(0\oplus g_{\eta})z_{\tilde{S} \boxplus \tilde{R}}(f_{\alpha}\oplus 0)
\end{equation}
\begin{equation}\label{IntRels6}
z_{\tilde{S} \boxplus \tilde{R}}(f_{\alpha}\oplus 0)z^*_{\tilde{S} \boxplus \tilde{R}}(0\oplus g_{\eta}) = z^*_{\tilde{S} \boxplus \tilde{R}}(0\oplus g_{\eta})z_{\tilde{S} \boxplus \tilde{R}}(f_{\alpha}\oplus 0)
\end{equation}\par
These operators along with the identity $1_{\tilde{\mathcal{H}}\oplus \tilde{\mathcal{K}}}$ generate the polynomial algebra $\mathcal{P}_{\tilde{S}\boxplus \tilde{R}}$ and form our natural Fock representation of the algebra $\mathcal{Z}(S \boxplus R, \mathcal{L})$. \par
We now consider a tensor product of Fock spaces
$$\mathcal{F}_{\tilde{S}}(\tilde{\mathcal{H}})\otimes \mathcal{F}_{\tilde{R}}(\tilde{\mathcal{K}}),$$
and on each Fock space we have creation/annihilation operators $z_{\tilde{S}}^*, z_{\tilde{S}}, z_{\tilde{R}}^*, z_{\tilde{R}}$ acting endomorphically on $\mathcal{F}_{\tilde{S}}^0(\tilde{\mathcal{H}}), \mathcal{F}_{\tilde{R}}^0(\tilde{\mathcal{K}})$ respectively, vacuum vectors $\Omega_{\tilde{S}}, \Omega_{\tilde{R}}$ define similar data on $\mathcal{F}_{\tilde{S}}(\tilde{\mathcal{H}})\otimes \mathcal{F}_{\tilde{R}}(\tilde{\mathcal{K}}):$
\begin{equation}\label{BoxSumOps1}
z_{\tilde{S}, \tilde{R}}\left(f_{\alpha} \oplus g_{\xi}\right) := z_{\tilde{S}, \alpha}(f)\otimes 1_{\tilde{\mathcal{K}}} + 1_{\tilde{\mathcal{H}}} \otimes z_{\tilde{R},\xi}(g), \quad f,g \in \mathcal{L},
\end{equation}
\begin{equation}\label{VacSum}
\Omega_{\tilde{S}, \tilde{R}} := \Omega_{\tilde{S}}\otimes \Omega_{\tilde{R}}.
\end{equation}
The polynomial algebra $\mathcal{P}_{\tilde{S}, \tilde{R}}$ is then defined as the algebra generated by the operators $1_{\tilde{\mathcal{H}}}\otimes 1_{\tilde{\mathcal{K}}}, z^*_{\tilde{S}, \tilde{R}}, z_{\tilde{S}, \tilde{R}}$.
\begin{lemma}\label{CyclicVec}
The vacuum vector $\Omega_{\tilde{S}, \tilde{R}}$ is cyclic for the polynomial algebra $\mathcal{P}_{\tilde{S}, \tilde{R}}$. That is, $\mathcal{P}_{\tilde{S}, \tilde{R}} \Omega_{\tilde{S}, \tilde{R}}$ is dense in $\mathcal{F}_{\tilde{S}}(\tilde{\mathcal{H}}) \otimes \mathcal{F}_{\tilde{R}}(\tilde{\mathcal{K}})$.
\end{lemma}
\begin{proof}
Let $\psi \in \mathcal{F}_{\tilde{S}}(\tilde{\mathcal{H}})\otimes \mathcal{F}_{\tilde{R}}(\tilde{\mathcal{K}})$ be orthogonal to $\mathcal{P}_{\tilde{S}, \tilde{R}}\Omega_{\tilde{S}, \tilde{R}}$. Then for any $i,j \in \mathbb{N}_0$ and vectors $f_1, \ldots, f_i \in \tilde{\mathcal{H}}, g_1, \ldots, g_j \in \tilde{\mathcal{K}}$
\begin{equation*}
\begin{split}
&\langle \psi_{i,j}, z^*_{\tilde{S},\tilde{R}}(f_1\oplus 0) \cdots z^*_{\tilde{S},\tilde{R}}(f_i\oplus 0)z^*_{\tilde{S},\tilde{R}}(0\oplus g_1)\cdots z^*_{\tilde{S},\tilde{R}}( 0\oplus g_j)\Omega_{\tilde{S}}\otimes\Omega_{\tilde{R}}\rangle \\
&= \sqrt{i! j!}\langle \psi_{i,j}, P_i^{\tilde{S}}\otimes P_j^{\tilde{R}}\left( f_1\otimes \cdots \otimes f_i\otimes g_1\otimes \cdots \otimes g_j\right) \rangle \\
&= \sqrt{i! j!}\langle \psi_{i,j}, f_1\otimes \cdots \otimes f_i\otimes g_1\otimes \cdots \otimes g_j \rangle 
\end{split}
\end{equation*}
where $\psi_{i,j}$ is the $i,j$-th component of $\psi$, each letter corresponding to each tensor slot in $\mathcal{F}_{\tilde{S}}(\tilde{\mathcal{H}}) \otimes \mathcal{F}_{\tilde{R}}(\tilde{\mathcal{K}})$ and we have used that the self-adjoint projection $P_i^{\tilde{S}}\otimes P_j^{\tilde{R}}$ leaves the vector $\psi$ invariant. By the definition of the tensor product, the vectors $ f_1\otimes \cdots \otimes f_i\otimes g_1\otimes \cdots \otimes g_j$ form a total set in $\tilde{\mathcal{H}}^{\otimes i} \otimes \tilde{\mathcal{K}}^{\otimes j}$ and hence we conclude that $\psi =0$.
\end{proof}
We are now ready to prove the claimed isomorphism of Fock spaces as representations of $\mathcal{Z}(S\boxplus R, \mathcal{L})$.
\begin{theorem}\label{MainThe}
Let $\mathcal{H}, \mathcal{K}$ be Hilbert spaces of finite dimensions $d_{\mathcal{H}}, d_{\mathcal{K}}$, respectively, and $\tilde{S}\in \mathcal{R}_0(\tilde{\mathcal{H}}), \tilde{R} \in \mathcal{R}_0(\tilde{\mathcal{K}})$, then:
\begin{enumerate}[a)]
\item The map $\pi_{\tilde{S}, \tilde{R}}: \mathcal{Z}(S\boxplus R, \mathcal{L}) \to \mathcal{P}_{\tilde{S}, \tilde{R}}$
\begin{equation*}
\begin{split}
\pi_{\tilde{S}, \tilde{R}}\left(1_{\mathcal{Z}(S\boxplus R, \mathcal{L})}\right)&:= 1_{\tilde{\mathcal{H}}\otimes \tilde{\mathcal{K}}},\\
 \pi_{\tilde{S}, \tilde{R}}(Z_{\alpha}(f))& := \begin{cases} z_{\tilde{S}, \tilde{R}}(f_{\alpha} \oplus 0), \ & \alpha\in \{1, \ldots, d_{\mathcal{H}}\}\\ z_{\tilde{S}, \tilde{R}}(0\oplus f_{\alpha -d_{\mathcal{H}}}), \ & \alpha\in \{d_{\mathcal{H}} +1, \ldots, d_{\mathcal{H}}+d_{\mathcal{K}}\}\end{cases}
\end{split}
\end{equation*}
extends to a unital $*$-representation of $\mathcal{Z}(S\boxplus R, \mathcal{L})$ on $\mathcal{F}^0_{\tilde{S}}(\tilde{\mathcal{H}}) \otimes \mathcal{F}^0_{\tilde{R}}(\tilde{\mathcal{K}})$ with cyclic vector $\Omega_{\tilde{S}, \tilde{R}}$ and
\begin{equation}\label{GNSInt}
\omega_{\tilde{S}, \tilde{R}}(X) = \langle \Omega_{\tilde{S}, \tilde{R}}, \pi_{\tilde{S}, \tilde{R}}(X)\Omega_{\tilde{S}, \tilde{R}}\rangle, \quad X\in \mathcal{Z}(S\boxplus R, \mathcal{L}).
\end{equation}
\item There exists a unitary $V: \mathcal{F}_{\tilde{S}\boxplus \tilde{R}}(\tilde{\mathcal{H}} \oplus \tilde{\mathcal{K}}) \to \mathcal{F}_{\tilde{S}}(\tilde{\mathcal{H}}) \otimes \mathcal{F}_{\tilde{R}}(\tilde{\mathcal{K}})$ such that
\begin{equation}\label{VUni}
V\Omega_{\tilde{S}\boxplus \tilde{R}} = \Omega_{\tilde{S}, \tilde{R}},\qquad  V\pi_{\tilde{S}\boxplus \tilde{R}}(X)V^* = \pi_{\tilde{S}, \tilde{R}}(X), \quad (X\in \mathcal{Z}(S\boxplus R, \mathcal{L})).
\end{equation}
\end{enumerate}
\end{theorem}
\begin{proof}
\begin{enumerate}[a)]
\item We show first that the operators $z_{\tilde{S}, \tilde{R}}^*, z_{\tilde{S}, \tilde{R}}$ satisfy the same relations as $z_{\tilde{S}\boxplus \tilde{R}}^*, z_{\tilde{S}\boxplus \tilde{R}}$ as outlined in \eqref{IntRels1}-\eqref{IntRels6}, firstly noting that
$$z_{\tilde{S}, \tilde{R}}(f_{\alpha}\oplus 0) = z_{\tilde{S},\alpha}(f)\otimes 1_{\tilde{\mathcal{K}}}, \quad z_{\tilde{S}, \tilde{R}}(0\oplus f_{\xi}) = 1_{\tilde{\mathcal{K}}}\otimes z_{\tilde{R}, \xi}(f), \quad (f\in \mathcal{L}),$$
and similarly for $z^*_{\tilde{S}, \tilde{R}}$.\par
Furthermore, the operators $z_{\tilde{S}}^*, z_{\tilde{S}}$ and $z_{\tilde{R}}^*, z_{\tilde{R}}$ satisfy exchange relations \eqref{InfRels1}, \eqref{InfRels2} governed by $S$ and $R$, respectively. Let $f, g \in \mathcal{L}$ then
\begin{equation*}
\begin{split}
z_{\tilde{S}, \tilde{R}}(f_{\alpha}\oplus 0)z_{\tilde{S}, \tilde{R}}(g_{\beta}\oplus 0)& = z_{\tilde{S},\alpha}(f)z_{\tilde{S}, \beta}(g) \otimes 1_{\tilde{\mathcal{K}}} \\
&= S^{\beta \alpha}_{\delta \gamma} z_{\tilde{S}, \gamma}(g)z_{\tilde{S},\delta}(f) \otimes 1_{\tilde{\mathcal{K}}} \\
&= S^{\beta \alpha}_{\delta \gamma} z_{\tilde{S}, \tilde{R}}(g_{\gamma}\oplus 0)z_{\tilde{S}, \tilde{R}}(f_{\delta}\oplus 0),
\end{split}
\end{equation*}
which gives \eqref{IntRels1}. The relation \eqref{IntRels3} follows in an analogous way on the second tensor factor applying \eqref{InfRels1} for $z_{\tilde{R}}$. Similarly
\begin{equation*}
\begin{split}
z_{\tilde{S}, \tilde{R}}(f_{\alpha}\oplus 0)z^*_{\tilde{S}, \tilde{R}}(g_{\beta}\oplus 0) &= z_{\tilde{S},\alpha}(f)z^*_{\tilde{S},\beta}(g) \otimes 1_{\tilde{\mathcal{K}}} \\
&= S^{\alpha \gamma}_{\beta \delta}z^*_{\tilde{S},\gamma}(g)z_{\tilde{S}, \delta}(f) \otimes 1_{\tilde{\mathcal{K}}} + \delta^{\alpha}_{\beta} \langle f,g\rangle 1_{\tilde{\mathcal{H}}}\otimes 1_{\tilde{\mathcal{K}}}\\
&= S^{\alpha \gamma}_{\beta \delta}z_{\tilde{S}, \tilde{R}}^*(g_{\gamma}\oplus 0)z_{\tilde{S}, \tilde{R}}(f_{\delta}\oplus 0) + \delta^{\alpha}_{\beta} \langle f,g\rangle 1_{\tilde{\mathcal{H}}}\otimes 1_{\tilde{\mathcal{K}}}.
\end{split}\end{equation*}
As for \eqref{IntRels5}, \eqref{IntRels6} we see that since $z_{\tilde{S}, \tilde{R}}(f_{\alpha} \oplus 0)$ and $z_{\tilde{S}, \tilde{R}}(0\oplus g_{\eta})$ operate on different tensor slots they commute. Cyclicity of the vacuum vector has been shown in Lemma \ref{CyclicVec} and we note that the annihilation of the normalised vector $\Omega_{\tilde{S}, \tilde{R}}$ by $z_{\tilde{S}, \tilde{R}}$ shows that the functional defined by the right hand side of \eqref{GNSInt} satisfies the properties listed in Definition \ref{FinFunDef} and therefore coincides with $\omega$. 
\item Let $X\in \mathcal{Z}(S\boxplus R, \mathcal{L})$. Then
\begin{equation*}
\begin{split}
\|  \pi_{\tilde{S},\tilde{R}}(X)\Omega_{\tilde{S}, \tilde{R}}\|^2 &= \langle \pi_{\tilde{S},\tilde{R}}(X)\Omega_{\tilde{S}, \tilde{R}},  \pi_{\tilde{S},\tilde{R}}(X)\Omega_{\tilde{S}, \tilde{R}} \rangle \\
&= \langle \Omega_{\tilde{S},\tilde{R}} ,  \pi_{\tilde{S},\tilde{R}}(X^*X)\Omega_{\tilde{S}, \tilde{R}} \rangle  \\
&= \omega(X^*X)\\
&= \| \pi_{\tilde{S}\boxplus \tilde{R}}(X)\Omega_{\tilde{S}\boxplus \tilde{R}} \|^2.
\end{split}
\end{equation*}
This shows that the map $V: \mathcal{P}_{\tilde{S}\boxplus \tilde{S}}\Omega_{\tilde{S}\boxplus\tilde{R}} \to \mathcal{P}_{\tilde{S}, \tilde{R}} \Omega_{\tilde{S}, \tilde{R}}$, $V\pi_{\tilde{S}\boxplus \tilde{R}}(X) \Omega_{\tilde{S}\boxplus \tilde{R}} := \pi_{\tilde{S}, \tilde{R}}(X) \Omega_{\tilde{S}, \tilde{R}}$ ($X\in\mathcal{Z}(S\boxplus R, \mathcal{L})$ is well-defined and isometric. Moreover, for $X, Y \in \mathcal{Z}(S\boxplus R, \mathcal{L})$
\begin{equation*}
\begin{split}
V\pi_{\tilde{S}\boxplus \tilde{R}}(X) V^* \pi_{\tilde{S}, \tilde{R}}(Y) \Omega_{\tilde{S}, \tilde{R}} &= V\pi_{\tilde{S}\boxplus \tilde{R}}(X) \pi_{\tilde{S}\boxplus \tilde{R}}(Y)  \Omega_{\tilde{S}\boxplus \tilde{R}}\\
&= V\pi_{\tilde{S}\boxplus \tilde{R}}(XY) \Omega_{\tilde{S}\boxplus \tilde{R}} \\
&=  \pi_{\tilde{S}, \tilde{R}}(XY) \Omega_{\tilde{S}, \tilde{R}}\\
& =  \pi_{\tilde{S}, \tilde{R}}(X) \pi_{\tilde{S}, \tilde{R}}(Y) \Omega_{\tilde{S}, \tilde{R}}.
\end{split}
\end{equation*}
Since $\Omega_{\tilde{S}\boxplus \tilde{R}}$ is cyclic for the representation $\pi_{\tilde{S}\boxplus \tilde{R}}$ and $\Omega_{\tilde{S}, \tilde{R}}$ is cyclic for $\pi_{\tilde{S}, \tilde{R}}$, it follows that $V$ extends to a unitary $\mathcal{F}_{\tilde{S}\boxplus \tilde{R}}(\tilde{\mathcal{H}}\oplus \tilde{\mathcal{K}}) \to \mathcal{F}_{\tilde{S}}(\tilde{\mathcal{H}}) \otimes \mathcal{F}_{\tilde{R}}(\tilde{\mathcal{K}})$ satisfying \eqref{VUni}.
\end{enumerate}
\end{proof}
As a corollary, we note the simple form of $S$-symmetrised Fock spaces in the case of $S$ a normal form as previously anticipated.
\begin{corollary}
Let $S\in R_0(\mathcal{H})$ be an involutive $R$-matrix of normal form \eqref{Normal}. Then there is a unitary $$V: \mathcal{F}_{\tilde{S}}(\tilde{\mathcal{H}}) \to \bigotimes_{i=1}^N \mathcal{F}_{\tilde{\varepsilon}_i}(\tilde{\mathcal{H}}_i),$$
where $\mathcal{H} = \bigoplus_{i=1}^N \mathcal{H}_i$.
\end{corollary}
By the results of \cite{Y-B}, we know that any involutive $R$-matrix is {\em equivalent to} an $R$-matrix of normal form. This motivates to ask what we can say about the relation of the Hilbert spaces and the representations of $\mathcal{Z}(S, \mathcal{L}), \mathcal{Z}(R, \mathcal{L})$ for two equivalent $R$-matrices $S\sim R$. If they are equivalent, we have unitary intertwiners $Y_n^{\tilde{S},\tilde{R}}$ \eqref{InterOps}, and their direct sum $Y^{\tilde{S}, \tilde{R}} = \bigoplus_n Y_n^{\tilde{S}, \tilde{R}}$ defines a unitary $\mathcal{F}_{\tilde{S}}(\tilde{\mathcal{H}}) \to \mathcal{F}_{\tilde{R}}(\tilde{\mathcal{K}})$.\par
In some cases, this Hilbert space isomorphism also intertwines the actions of the Zamolodchikov operators as we now discuss in two examples.
\begin{example}\label{EquivThe}
Let $\tilde{S}\in \mathcal{R}_0(\tilde{\mathcal{H}}), \tilde{R}\in \mathcal{R}_0(\tilde{\mathcal{K}})$ with $S\sim R$ (in the type 1 sense) such that we may choose their intertwiner to be of the form
\begin{equation}\label{StrongInter}
Y_n^{S,R} = \left(Y_1^{S,R}\right)^{\otimes n}.
\end{equation}
Then 
\begin{equation}\label{StrongIso}
Y^{\tilde{S},\tilde{R}}\pi_{\tilde{S}}(Z_{\alpha}(f))\left(Y^{\tilde{S},\tilde{R}}\right)^* = \pi_{\tilde{R}}\left(Y_1^{\tilde{S},\tilde{R}}Z_{\alpha}(f)\right).
\end{equation}
\par
Indeed, we can check by calculation -  Let $f \in \mathcal{L}$, then by Corollary \ref{Coro1} we have the expression $Y^{\tilde{S}, \tilde{R}} = U^*\left(Y^{S,R} \otimes 1\right) U$ where we note that the intertwiner acts trivially on $\mathcal{L}$. This gives
\begin{equation*}
\begin{split}
&U^*\left(Y^{S,R} \otimes 1\right)U z_{\tilde{S}, \alpha}(f)U^*\left(\left(Y^{S,R}\right)^* \otimes 1\right)U \\
&\qquad \qquad  =U^*\left(Y^{S,R} \otimes 1\right) z_S(e_{\alpha}) \otimes z_F(f) \left(\left(Y^{S,R}\right)^* \otimes 1\right)U\\
& \qquad \qquad = U^* \left( Y^{S,R}z_S(e_{\alpha})\left(Y^{S,R}\right)^* \otimes z_F(f) \right) U
\end{split}
\end{equation*}
What remains to be shown is to calculate the action of $Y^{S,R}z_S(e_{\alpha})\left(Y^{S,R}\right)^*$. Let $\psi \in \mathcal{F}_S(\mathcal{H})$ then
\begin{equation*}
\begin{split}
Y^{S,R}_{n+1}z^*_S(e_{\alpha})\left(Y^{S,R}_n\right)^*\psi_n &= Y^{S,R}_{n+1} P_{n+1}^S e_{\alpha} \otimes \left(\left[Y_1^{S,R}\right]^*\right)^{\otimes n}\psi_n \\
&= P_{n+1}^R \left(Y_{1}^{S,R}\right)^{\otimes (n+1)} \left(e_{\alpha} \otimes \left(\left[Y_1^{S,R}\right]^*\right)^{\otimes n}\psi_n\right) \\
&= P_{n+1}^R\left( Y_1^{S,R}e_{\alpha} \right)\otimes \psi_n\\
&= z^*_R\left(Y_1^{S,R} e_{\alpha}\right) \psi_n.
\end{split}
\end{equation*}
Thus the adjoint action of $Y^{\tilde{S}, \tilde{R}}$ results in an isomorphism between elements of the polynomial algebras $\mathcal{P}_{\tilde{S}}, \mathcal{P}_{\tilde{R}}$.
\end{example}
\par
It must be mentioned, however, that the isomorphism of the Hilbert spaces given by $Y^{\tilde{S}, \tilde{R}}$ does {\em not} always give an isomorphism of Zamolodchikov representations. As a counter example, we restrict to dimension two and consider $S= -F_{\mathcal{H}} \sim -1\boxplus -1 = R$. In the Fock representation of $\mathcal{Z}(S, \mathcal{L})$ we have anti-commutation between all annihilation operators, i.e.
$$\left\{ z_{\tilde{S}, \alpha}(f) , z_{\tilde{S}, \beta}(g) \right\} = 0$$
for $\alpha, \beta \in \{1, 2\}$ and all $f,g \in \mathcal{L}$. If an isomorphism between the representations of $\mathcal{Z}(S, \mathcal{L})$ and $\mathcal{Z}(R, \mathcal{L})$ existed, this anti-commutation would be preserved, however for some choices of $\alpha$ and $\beta$ we actually have commutation in the representation of $\mathcal{Z}(R, \mathcal{L})$
$$\left[ z_{\tilde{R}, 1}(f), z_{\tilde{R}, 2}(g) \right] =0$$
for all $f,g \in \mathcal{L}$. A quick calculation shows the product $z_{\tilde{R}, 1}(f) z_{\tilde{R}, 2}(g)$ is not zero and thus the representations $\pi_{\tilde{S}}, \pi_{\tilde{R}}$ are not isomorphic in this case.

\par
\section{Outlook to Quantum Field Theory}\label{section:outlook}
Our work is largely motivated by applications to short distance scaling limits of integrable quantum field theories on two-dimensional Minkowski space, which we sketch now.\par
At finite scale, such quantum field theories can be described in terms of a mass parameter $m>0$ and an $S$-matrix (see Definition \ref{SMatrix} below) that describes $2\to 2$ collision processes in this model \cite{Lechner}. This $S$-matrix depends on a parameter $\theta$, the rapidity difference of the two scattering particles, which is related to the on-shell momentum by
$$p(\theta) = m\begin{pmatrix} \cosh(\theta) \\ \sinh(\theta) \end{pmatrix}.$$
Simplifying slightly from the general situation \cite{5}, we assume that our model contains only a single species of massive neutral particles. Then the appropriate definition of an $S$-matrix is the following (with the usual notation $S(a,b):=\{\zeta\in{\mathbb C}\,:\,a<\Im\zeta<b\}$ for open strips in the complex plane).
\begin{definition}\label{SMatrix}
An $S$-matrix is a continuous bounded function $S: \overline{S(0,\pi)} \to \mathcal{B}(\mathcal{H}\otimes \mathcal{H})$ which is analytic in the interior of the strip and satisfies for all $\theta, \theta' \in \mathbb{R},$ and $\alpha, \beta, \delta, \gamma \in \{1,\ldots, d_{\mathcal{H}}\}$,
\begin{enumerate}[a)]
\item Unitarity:
$$S(\theta)^* = S(\theta)^{-1},$$
\item Hermitian Analyticity:
$$S(\theta)^{-1} = S(-\theta),$$
\item Solution to the Yang-Baxter equation:
$$(S(\theta) \otimes 1_{\mathcal{H}})(1_{\mathcal{H}} \otimes S(\theta +\theta')(S(\theta') \otimes 1_{\mathcal{H}}) = (1_{\mathcal{H}} \otimes S(\theta'))(S(\theta + \theta') \otimes 1_{\mathcal{H}})(1_{\mathcal{H}} \otimes S(\theta)),$$
\item Crossing symmetry:
$$S^{\alpha \beta}_{\delta \gamma}(\theta) = S^{\delta \alpha}_{\gamma \beta}(i\pi - \theta).$$
\end{enumerate}
\end{definition}
\par
We recall that to any mass $m>0$ and any $S$-matrix $S$ satisfying certain regularity and intertwiner properties, a corresponding QFT exists that has $S$ as its 2-particle $S$-matrix \cite{Alazzawi:2016set, SabArt}. The vacuum Hilbert space of this model is defined as an $S$-symmetric Fock space over the one-particle space $\tilde{\mathcal H}={\mathcal H}\otimes\mathcal L$ with ${\mathcal L}=L^2({\mathbb R},d\theta)$, and the defining (wedge-local) quantum fields are sums of Zamolodchikov creation and annihilation operators. This is analogous to our procedure in Section 2 up to the essential difference that $S$ does not only act on $\mathcal H$ due to its dependence on the rapidity $\theta$.

\par
Proceeding to the scaling limit, one can generalise the analysis of \cite{Scaling} to show that the existence of a short distance scaling limit requires the two limits $$S_\pm := \lim_{\theta \to \pm \infty}S(\theta) $$ to exist \cite{Me}. This is, however, only a necessary condition -- It is not yet clear which additional assumptions $S$ must satisfy to guarantee that the scaling limit theory has physically reasonable properties, or which of these scaling limits are isomorphic to interaction-free theories (asymptotic freedom).

As expected of a short distance limit, the scaling limit theory is massless and (twisted) chiral. Some general properties of this limit theory can be described in terms of the matrices $S_0:= S(0)$ and $S_\pm$. In the context of the present paper, it is interesting to realise that these three matrices do no longer depend on $\theta$ and hence are unitary solutions of the Yang-Baxter equation \eqref{Y-BC} without spectral parameter. Furthermore, in view of Hermitian Analyticity, $S_0$ is involutive (that is, $S_0\in\mathcal{R}_0(\mathcal{H})$) and $S_+^*=S_-$.

The physical significance of $S_0,S_{\pm}$ can be described as follows \cite{Me}. The role of $S_0$ is to describe a possible twisting (commutation/anticommutation relations) of the two chiral halves of the scaling limit theory. The limits $S_\pm$, on the other hand, present obstructions to the existence of non-trivial operators localized in (lightlike) intervals on the two chiral halves. These obstructions are given either by representations of the symmetric groups (if $S_+=S_-$ is involutive), or, more generally, representations of the braid groups (if $S_+\neq S_-$ is not involutive).

The scalar case, i.e. $\dim(\mathcal{H})=1$, was analysed in \cite{Scaling}. In that setting, the only possible limits are involutive, namely $S_+=S_-$ coincides with either $+1$ or $-1$. In case the (scalar) S-matrix is independent of $\theta$, that is $S(\theta)=\pm1$ for all $\theta\in\mathbb{R}$, also the scaling limit theory is completely understood (it is the free $U(1)$ current for $S=+1$ and related to the Ising model for $S=-1$). 

In the general non-scalar case considered here, several natural questions come up: 1)~Do the limits $S_+$ and $S_-$ necessarily coincide? (which would mean that they are involutive) 2)~What are all possible limits of $S_\pm$ of S-matrices? 3)~What are all constant ($\theta$-independent) solutions of Definition~\ref{SMatrix}?

\par
\vspace{5mm}

Regarding 2) and 3), note that the condition of crossing symmetry has not appeared in our previous discussions. We will call a constant $R$-matrix crossing-symmetric if
$$S^{\alpha \beta}_{\delta \gamma} = S^{\delta \alpha}_{\gamma \beta}$$
for all $\alpha,\beta,\gamma,\delta$ labeling an orthonormal basis of $\mathcal H$ (see also \cite{Lyubashenko:1987} for a basis-independent formulation of this concept).

In the remainder of this paper we will discuss some example classes of S-matrices: The first class are the involutive normal forms \eqref{Normal} as a set of possible constant $S$-matrices, and we study their suitability as candidates for limits of certain parameter-dependent $S$-matrices. The second class are the so-called diagonal $S$-matrices, which we will show to always have coinciding limits $S_+=S_-$, and the last class are all constant R-matrices in dimension 2, for which the condition of crossing symmetry can be understood completely.
\begin{proposition}
Let $S\in \mathcal{R}_0(\mathcal{H})$ be of normal form:
$$S= \bigbplus_{i=1}^N\  \varepsilon_i 1_{d_i^2}$$
for some $N\in \mathbb{N}$, each $\varepsilon_i \in \{+1, -1\}$ and $d_1+\cdots +d_N = d_{\mathcal{H}}$. Then $S$ is crossing symmetric if and only if $S$ is of diagonal form, i.e. $d_i=1$ for all $i$.
\end{proposition}
\begin{proof}
``If": If $S$ is a diagonal normal form, the only non-zero entries are ones $S^{\alpha \alpha}_{\alpha \alpha}$ and $S^{\alpha \beta}_{\beta \alpha} = S^{\beta \alpha}_{\alpha \beta}$ by definition, so crossing symmetry is easily realised.\par
``Only if'': Suppose that $d_i>1$ for some $i$. Then there is a subspace $V$ of $\mathcal H$, of dimension $\dim V>1$, such that $R$ acts as $\pm\text{id}_{V\otimes V}$ on this subspace. But the identity in dimension larger one is not crossing symmetric. In fact, its ``crossing partner'' $\hat{1}^{\alpha\beta}_{\gamma\delta}:=1^{\delta\alpha}_{\gamma\beta}=\delta^{\delta}_\gamma\delta^\alpha_\beta$ is not even invertible.
\end{proof}

This results motivates us to consider $\theta$-dependent diagonal S-matrices. We first introduce two suitable classes of holomorphic functions.

\begin{definition}
\begin{enumerate}[i)]
\item A bounded continuous function $f: \overline{S(0,\pi)} \to \mathbb{C}$ which is analytic in the interior of the strip is said to be {\em regular} if it extends to a bounded analytic function on the open strip $S(-\kappa, \pi+ \kappa)$ for some $\kappa >0$.
\item The set of functions $\mathcal{G}_{\text{lim}}$ is the set of regular functions $G:\overline{S(0,\pi)} \to \mathbb{C}$ that satisfy for all $\theta \in \mathbb{R}$:
\begin{equation}
\label{GProps}
|G(\theta)|=1,\ G(\theta) = \overline{G(i\pi +\theta)}, \ \lim_{\theta \to \pm \infty} G(\theta) \ \text{exist}.\end{equation}
\end{enumerate}
\end{definition}
\begin{remark} Note that with the additional symmetry 
$$G(-\theta)= \overline{G(\theta)}$$
the function $G$ actually belongs to a special class of functions known as {\em scattering functions}. The class of all scattering functions with limits $\lim_{\theta \to \pm \infty} G(\theta)$ is explicitly known as certain symmetric finite Blaschke products \cite{Scaling}. We are interested in the more general class $\mathcal{G}_{\text{lim}}$ because these functions appear as matrix elements of diagonal S-matrices.
\end{remark}
\begin{definition}
A \textit{diagonal S-matrix with limits} $S_D$ is of the form
\begin{equation}
\label{DiagonalMatrix}
S_D(\theta)^{\alpha \beta}_{\gamma \delta} = \omega_{\alpha \beta}(\theta) \delta^{\alpha}_{\delta}\delta^{\beta}_{\gamma},
\end{equation}
(with no summation over $\alpha$ and $\beta$) with $\omega_{\alpha \beta} \in \mathcal{G}_{\text{lim}}$ for all $\alpha, \beta$ and $\omega_{\alpha \alpha}(-\theta) = \overline{\omega_{\alpha \alpha}(\theta)}$.
\end{definition}
\begin{remark}
It is well known that diagonal S-matrices satisfy all requirements of Definition~\ref{SMatrix} \cite{Alazzawi:2016set}, so we can treat them as a simple example of an S-matrix without restricting ourselves to low dimensions.
\end{remark}
\par

\begin{proposition}
The set $\mathcal{G}_{\textup{lim}}$ consists precisely of the functions
\begin{equation}\label{GFunction}
G(\zeta) = \epsilon \prod_{k=1}^N \frac{e^{\zeta} - e^{z_k}}{e^{\zeta}-e^{\overline{z_k}}}\cdot \frac{e^{\zeta} - e^{\overline{z_k}+i\pi}}{e^{\zeta} - e^{z_k-i\pi}},
\end{equation}
where $\epsilon =\pm 1$, $N\in \mathbb{N}_0$ and $\{ z_1, \ldots, z_N\}$ is a set of complex numbers in the strip $0< \Im z_1,\ldots, \Im z_N \leq \pi/2$.\par
Moreover, for each $G\in \mathcal{G}_{\text{lim}}$ we have that $$\lim_{\theta \to \infty} G(\theta) = \lim_{\theta \to- \infty} G(\theta) = \pm 1.$$
\end{proposition}
\begin{proof}
In the following we will use $\zeta \in \mathbb{C}$ to denote complex arguments and $\theta \in \mathbb{R}$ to denote real arguments in functions. \par
Each factor 
$$g_{z_k}: \zeta \mapsto \pm \frac{e^{\zeta} - e^{z_k}}{e^{\zeta}-e^{\overline{z_k}}}\cdot \frac{e^{\zeta} - e^{\overline{z_k}+i\pi}}{e^{\zeta} - e^{z_k-i\pi}}$$
clearly satisfies $|g_{z_k}(\theta)| =1$  and also $g_{z_k}(\theta) = \overline{g_{z_k}(i\pi +\theta)}$. Indeed,
\begin{equation*}
\begin{split}
g_{z_k}(i\pi +\theta) &=\pm \frac{-e^{\theta} - e^{z_k}}{-e^{\theta}-e^{\overline{z_k}}}\cdot \frac{-e^{\theta} - e^{\overline{z_k}+i\pi}}{-e^{\theta} - e^{z_k-i\pi}}\\
&= \pm \frac{e^{\theta} + e^{z_k}}{e^{\theta}+e^{\overline{z_k}}}\cdot \frac{e^{\theta} + e^{\overline{z_k}+i\pi}}{e^{\theta} + e^{z_k-i\pi}}\\
&= \pm \frac{e^{\theta} - e^{z_k-i\pi}}{e^{\theta}-e^{\overline{z_k}+i\pi}}\cdot \frac{e^{\theta} - e^{\overline{z_k}}}{e^{\theta} - e^{z_k}} = \overline{g_{z_k}(\theta)}.
\end{split}
\end{equation*}
The location of the poles in the expression implies that for a sufficiently small $\delta >0$, the factor $g_{z_k}$ is analytic and bounded in the strip $S(-\Im z_k+\delta, \pi + \Im z_k -\delta) \supset S(0,\pi)$ and moreover the product \eqref{GFunction} is finite so it follows that $G$ is analytic and bounded in the strip $S(-\kappa, \pi + \kappa)$ for some $\kappa >0$. From \eqref{GFunction} it is clear that the limits $\lim_{\theta \to \infty} G(\theta), \lim_{\theta \to -\infty} G(\theta)$ exist and coincide with $\epsilon$ which shows that $G\in \mathcal{G}_{\text{lim}}$.\par
\vspace{2mm}
Now we pick an arbitrary $G\in \mathcal{G}_{\text{lim}}$ and show that it is of the form \eqref{GFunction}. Let $\epsilon_1 := \lim_{\theta\to \infty} G(\theta), \epsilon_2:= \lim_{\theta \to -\infty} G(\theta)$, then from the regularity properties of $G$, we have that $G(\theta + i\lambda) \to \epsilon_1$ as $\theta \to \infty$ uniformly in $\lambda \in [0,\pi]$ \cite[pp.~170]{titchmarsh1939theory}. Since $|G(\theta)| =1$ for all real $\theta$, we have $|\epsilon_1| = 1$. Moreover, $\overline{\epsilon_1}=\lim_{\theta \to \infty} \overline{G(\theta)} = \lim_{\theta \to \infty} G(i\pi + \theta)=\epsilon_1$, i.e. $\epsilon_1=\pm1$. Analogously one sees $\epsilon_2=\pm1$.

Furthermore, $G$ is continuous on the closed strip $\overline{S(0,\pi )}$ and of unit modulus on the boundary, so the uniform approach to the limits implies that $G$ has only finitely many zeroes in $S(0,\pi)$. Let $z_1, \ldots, z_N$ be the zeroes of $G$ whose imaginary parts $\mu$ satisfy $0<\mu \leq \frac{\pi}{2}$. For every zero $z_i$, there is also a corresponding zero $\overline{z_i} + i\pi$. Consider now the Blaschke product
$$B(\zeta) = \epsilon_1\prod_{k=1}^{N} \frac{e^{\zeta}- e^{z_k}}{e^{\zeta} - e^{\overline{z_k}}}\cdot \frac{e^{\zeta} - e^{\overline{z_k} + i\pi}}{e^{\zeta} - e^{z_k - i\pi}}.$$
Now, $B$ has precisely the same number of zeroes as $G$ and also $B(\theta + i\lambda) \to \epsilon_1$ for $\theta \to \pm \infty$ uniformly in $\lambda \in [0,\pi]$.
\par
Define a new function $F$ by
$$F = G\cdot B^{-1}.$$
By construction, $F$ is analytic and non-vanishing in $S(0,\pi)$ and also $F(\theta + i\lambda) \to 1$ for $\theta \to \infty$ and $F(\theta + i\lambda) \to \epsilon_2/\epsilon_1$ for $\theta \to -\infty$, uniformly for $\lambda \in [0,\pi]$. Thus $F$ is bounded above and below, and there exists a $K>0$ such that $K< |F(\theta)| <1$ for $\theta \in \overline{S(0,\pi)}$. Define now $F$ on the lower strip $S(-\pi,0)$ by
$$F(\zeta) = \frac{1}{\overline{F(\overline{\zeta})}},\qquad \zeta\in S(-\pi,0).$$ Since $F$ has unit modulus on the real line, this provides an analytic continuation of $F$ to $S(-\pi,\pi)$ which is uniformly bounded because $|F|$ was bounded above and below on $S(0,\pi)$. 
At the lower boundary we find, $\theta\in\mathbb R$, 
\begin{align*}
	F(\theta-i\pi)=\overline{F(\theta+i\pi)}^{-1}=F(\theta)^{-1}=F(\theta+i\pi).
\end{align*}
Hence $F$ extends to a $2\pi i$-periodic, entire function which is bounded and so constant by Liouville's theorem. Thus we have that $F(\theta) = \lim_{\theta \to \infty} F(\theta) = 1$, so $G = FB=B$.
\end{proof}
\begin{corollary}
Let $S_D(\theta)$ be a diagonal $S$-matrix with limits, then 
$$\lim_{\theta \to \infty}S_D(\theta) = \lim_{\theta \to -\infty}S_D(\theta)$$
and the limits are involutive diagonal R-matrices.
\end{corollary}
\begin{proof}
As all functions in $\mathcal{G}_{\text{lim}}$ have coinciding limits, this follows immediately.
\end{proof}
\par
The possible limits of these diagonal type $S$-matrices are of the form
$$
\left(\bigbplus_{i=1}^j \varepsilon_i\right) \boxplus \left(-\bigbplus_{i=1}^{d_{\mathcal{H}}-j}\varepsilon_i\right)$$
which are clearly involutive (and hence crossing symmetric, also) so the obstructions to local operators in the scaling limit are given by the permutation group rather than the braid group. It is possible that this is a more general phenomenon that is not restricted to the diagonal case.
\par

To conclude, let us consider constant not necessarily involutive (but unitary) R-matrices in dimension $2$. Strengthening a result of Dye \cite{Dye}, it was recently shown that 

in dimension 2, every such R-matrix is type 1 equivalent (see p.~\pageref{type1}) to one of the following four cases \cite{ContiLechner:2019}:

\begin{equation*}
\begin{split}
R_1 &= q\cdot 1, \quad q\in \mathbb{T}=\{ z\in \mathbb{C} \ | \ |z| =1\}\\
R_2 &= \begin{pmatrix} p&0&0&0\\ 0&0&q&0\\0&r&0&0\\0&0&0&s \end{pmatrix}, \quad p,q,r,s\in \mathbb{T},\\
R_3 &= \begin{pmatrix} 0&0&0&p\\0&q&0&0\\0&0&q&0\\r&0&0&0\end{pmatrix}, q, p\cdot r \in \mathbb{T},\\
R_4 &= \frac{q}{\sqrt{2}} \begin{pmatrix}1&1&0&0\\-1&1&0&0\\0&0&1&-1\\0&0&1&1\end{pmatrix}, \quad q\in \mathbb{T}.
\end{split}
\end{equation*}
We note here that the representatives $R_1, R_4$ are never crossing symmetric, $R_2$ is crossing symmetry if and only if $q=r$, and $R_3$ is crossing symmetric if and only if $p=q=r$, that is involutive up to a constant factor. Crossing symmetry also dictates that if $R_i$ is crossing symmetric, then $(Q\otimes Q)R_i(Q\otimes Q)^{-1}$ is crossing symmetry if and only if the unitary $Q$ is real-valued and hence orthogonal \cite{Me}. 

The interplay of crossing symmetric S-matrices, their (involutive) limits and the corresponding localization properties of the scaling limit theories will be investigated in detail in a future work.

\end{document}